\documentclass[preprint,12pt]{imsart}

\RequirePackage[numbers]{natbib}
\usepackage[english]{babel}
\usepackage{amsmath}
\usepackage{amssymb}
\usepackage{ulem}
\usepackage{natbib}
\usepackage{comment}
\usepackage{url}
\usepackage{color}
\usepackage{hyperref}
\usepackage[pdftex]{graphicx}
\usepackage[right = 2cm, top=2cm, bottom=3cm,left=2cm]{geometry}
\usepackage{mathdots}
\usepackage{amsthm}
\usepackage{bbm}
\usepackage{enumitem}
\usepackage{array,multirow}
\usepackage{xr}

\usepackage{framed}

\usepackage[T1]{fontenc}
\usepackage[utf8x]{inputenc}

\DeclareMathOperator{\argmax}{argmax}
\DeclareMathOperator{\argmin}{argmin}

\begin{document}

%\tableofcontents

\setlength\parindent{0pt}

\newcommand{\Rho}{P}
\newcommand{\IN}{\mathbb{N}}
\newcommand{\IQ}{\mathbb{Q}}
\newcommand{\IZ}{\mathbb{Z}}
\newcommand{\IR}{\mathbb{R}}
\newcommand{\IC}{\mathbb{C}}
\newcommand{\Ima}{\mbox{Im}}
\newcommand{\dif}{\ \mbox{d}}
\newcommand{\cov}{\mbox{cov}}

\newcommand{\Lp}{\mathcal{L}}
\newcommand{\sI}{\mathcal{I}}
\newcommand{\sA}{\mathcal{A}}
\newcommand{\sB}{\mathcal{B}}
\newcommand{\sP}{\mathcal{P}}
\newcommand{\sE}{\mathcal{E}}
\newcommand{\sF}{\mathcal{F}}
\newcommand{\sL}{\mathcal{L}}
\newcommand{\sG}{\mathcal{G}}
\newcommand{\sH}{\mathcal{H}}
\newcommand{\sT}{\mathcal{T}}
\newcommand{\sV}{\mathcal{V}}

\renewcommand{\Re}{\mbox{Re }}
\renewcommand{\Im}{\mbox{Im }}

\newcommand{\reff}[1]{(\ref{#1})}

\newcommand{\IP}{\mathbb{P}}
\newcommand{\IE}{\mathbb{E}}
\newcommand{\Ii}{\mathbbm{1}}
\newcommand{\supp}{\mbox{supp}}
\newcommand{\Hess}{\mbox{Hess}}
\newcommand{\Var}{\mbox{Var}}
\newcommand{\sX}{\mathcal{X}}
\newcommand{\Kov}{\mbox{Kov}}
\newcommand{\Cov}{\mbox{Cov}}
\newcommand{\tr}{\mbox{tr}}
\newcommand{\gdw}{\Leftrightarrow}
\newcommand{\pto}{\overset{p}{\to}}
\newcommand{\fsto}{\overset{a.s.}{\to}}
\newcommand{\dto}{\overset{d}{\to}}
\newcommand{\lto}{\overset{L^2}{\to}}
\newcommand{\sD}{\mathcal{D}}
\newcommand{\iid}{\overset{\mbox{iid}}{\sim}}
\renewcommand{\l}{\ell}

\newcommand{\lima}{\mbox{l.i.m.}}

\renewcommand{\supp}{\text{supp}}

\newcommand{\err}{\mbox{err}}
\newcommand{\bias}{\mbox{bias}}

\newcommand{\norm}[1]{\left\lVert#1\right\rVert}

\newtheorem{theorem}{Theorem}[section]
\newtheorem{corollary}[theorem]{Corollary}
\newtheorem{definition}[theorem]{Definition}
\newtheorem{proposition}[theorem]{Proposition}
\newtheorem{lemma}[theorem]{Lemma}
\newtheorem{remark}[theorem]{Remark}
\newtheorem{example}[theorem]{Example}
\newtheorem{assumption}[theorem]{Assumption}

\newcommand{\note}[1]
{$^{(!)}$\marginpar[{\hfill\tiny{\sf{#1}}}]{\tiny{\sf{(!) #1}}}}

%%%%%%%%%%%%%%%%%%%%%%%%%%%%
%%%%%%%%%%%TITLE START%%%%%%%%%%
%%%%%%%%%%%%%%%%%%%%%%%%%%%%

\begin{frontmatter}

\title{A supreme test for periodic explosive GARCH }
\runtitle{ GARCH detection}

\begin{aug}
  
  \author{\fnms{Stefan} \snm{Richter}\thanksref{a,e1}\ead[label=e1,mark]{stefan.richter@iwr.uni-heidelberg.de}},
  \author{\fnms{Weining}  \snm{Wang}\thanksref{b,e2}\ead[label=e2,mark]{weining.wang@city.ac.uk}}
  \and
  \author{\fnms{Wei Biao}  \snm{Wu}\thanksref{c,e3}\ead[label=e3,mark]{wbwu@galton.uchicago.edu}}%

  \runauthor{Richter, Wang and Wu}

  \affiliation{Heidelberg University and University of Chicago}

  \address[a]{Institut f\"{u}r Angewandte Mathematik, Universit\"{a}t Heidelberg, Im Neuenheimer Feld 205, 69120 Heidelberg, Germany. \printead{e1}}
  
  \address[b]{Department of Economics, City, U of London, center of applied statistics, HU Berlin, Unter den Linden 6, 10099 Berlin, Germany.  \printead{e2}}

  \address[c]{Department of Statistics, University of Chicago, 5734 S. University Avenue, Chicago, IL 60637, USA.  \printead{e3}}

\end{aug}

\begin{abstract}
We develop a uniform test for detecting  and dating explosive behavior of a strictly stationary GARCH$(r,s)$ (generalized autoregressive conditional heteroskedasticity)  process. Namely, we test the null hypothesis of a globally stable GARCH process with constant parameters against an alternative where there is an 'abnormal' period with changed parameter values. During this period, the change may lead to an explosive behavior of the volatility process. It is assumed that both the magnitude and the timing of the breaks are unknown.
We develop a double supreme test for the existence of a break, and then provide an algorithm to identify the period of change.
Our theoretical results hold under mild moment assumptions on the innovations of the GARCH process. Technically, the existing properties for the QMLE in the GARCH model need to be reinvestigated to hold uniformly over all possible periods of change. The key results involve a uniform weak Bahadur representation for the estimated parameters, which leads to weak convergence of the test statistic to the supreme of a Gaussian Process.
In simulations we show that the test has good size and power for reasonably large time series lengths. We apply the test to Apple asset returns and Bitcoin returns. \end{abstract}


\end{frontmatter}

\section{Introduction}
Volatility is an important indicator for economic and financial stability.
There is a growing evidence for the unstable behavior of the historical volatility of numerous micro level as well as macro level  data, such as individual stocks prices, asset returns, VIX, inflation and unemployment. \cite{bloom2007uncertainty} document the unstable behavior of higher moments of many economic variables, such as R\&D (research and develop) rates related to the uncertainty about future productivity. It is understood that the nature of uncertainty is the unpredictability of any model to the future path of a time series. Therefore it may be connected with a change of the parameter values in the underlying data generating process. A direct empirical fact is that one often sees a sudden explosive behavior in the second moment of the process which bounces back after a while.
For example, in Figure \ref{figure:bitcoin} we have plotted a realization of a piece-wise explosive GARCH(1,1) process and the log returns of Bitcoin. The whole time span is set to be (July 28, 2010- May 14, 2011) with an explosive period with changing parameter values (December 30, 2010 - October 15, 2011). The two trajectories of the time series look rather similar.
Such kind of data phenomena suggest that the underlying processes have time varying parameters, calling for a rigorous econometric treatment for detection of change periods and corresponding inference. 
 We see that the piecewise explosive GARCH process in Figure \ref{figure:bitcoin} captures the explosive behavior of Bitcoin in the squared returns.  The aim of our paper is to develop a generalised supreme test for GARCH models which is able to detect exuberance behaviour periods (periods with explosive parameter values) which are associated with the empirical phenomena of explosiveness in the second moment. \\

\begin{figure}[h] \label{bitcoin}
\caption{Plot of simulated piecewise explosive squared returns (upper panel) versus Bitcoin squared log returns (lower panel). A stable period (20100728- 20110514) versus an explosive period (20100901-20110301). }
%3.0000000 0.4376157 0.7081921
% 3.0000000 0.1011351 0.8788099
\centering
\label{figure:bitcoin}
\includegraphics[width=0.5\textwidth]{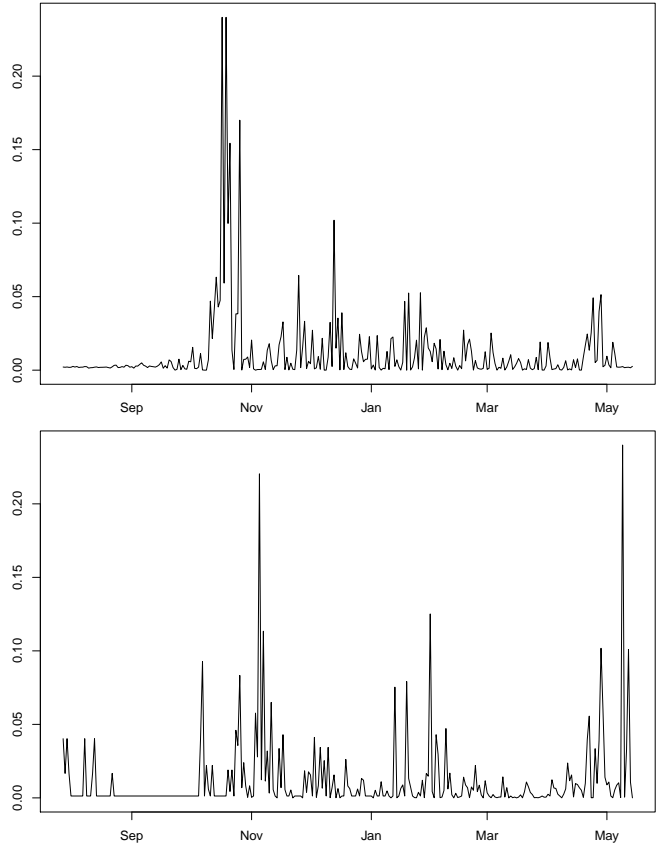}
\end{figure}

%There has been many paper on bubble testing regarding mildly explosive process.
%There is little econometric tools on dating explosive behavior of a GARCH models.
%This paper exploits a data rich environment to provide direct econometric estimates of time-varying macroeconomic uncertainty, 
%The unstable behavior is typically known to countercylical and might induced by uncertainty, 

%the quantitative phenomena is corresponding to a period explosive in the second moment of a process.\\

The highly celebrated autoregressive conditional heteroskedasticity (ARCH) model proposed by
\cite{engle1982autoregressive} is important for describing the pervasive phenomena of heteroskedasticity presented in many time series. One key generalization of ARCH is the GARCH model, i.e.
\begin{eqnarray}
    X_i^2 &=& \zeta_i^2 \sigma_i^2,\nonumber\\
    \sigma_i^2 &=& \alpha_0 + \sum_{j=1}^{r}\alpha_j X_{i-j}^2 + \sum_{k=1}^{s}\beta_k \sigma_{i-k}^2,\label{eq:garchrecursion}
\end{eqnarray}
where the conditional variance $\sigma_i^2$ depends on the past observations $X_{i-j}^2$ but also on the historical conditional variance $\sigma_{i-k}^2$. $\zeta_i$ are assumed to be i.i.d. innovations, see for example \cite{Pao:18} for more details of the model.

 Numerous estimation methods for the parameters of GARCH models have been proposed, and the consistency and asymptotic normality has been carefully studied in the literature. A conventional estimation approach is the quasi maximum likelihood estimation (QMLE), e.g. \cite{bollerslev1992quasi}. Also 
 \cite{fan2014quasi} study quasi-maximum likelihood estimation of GARCH models with heavy-tailed likelihoods. \cite{peng2003least} propose a least absolute deviation estimator.  \cite{jensen2004asymptotic} establish consistency and asymptotic normality of the quasi-maximum likelihood estimator in the linear ARCH model.
It is well known that under the assumption of strict stationarity of a GARCH model, there is still a region of parameter values allowing for realizations with unstable volatility behavior. The leading case is the ``IGARCH" process.  \cite{1990stationarity} looks at the behavior of an  ``IGARCH" process, and it is known that the unconditional mean of the IGARCH's conditional variance is not finite, which implies infinite second or higher moments (i.e. eruptive behavior). \cite{lee1994asymptotic} provide
an asymptotic theory for a strictly stationary GARCH(1,1) quasi-maximum likelihood estimator allowing for the case of IGARCH and explosive conditional variance and even nonstationarity.  \cite{jensen2004asymptotic} consider asymptotic inference for nonstationary GARCH model.

Despite the rich empirical literature which suggests the existence of an unstable moment period of a GARCH process, there is only sparse literature on determining and testing the period of explosiveness.  \cite{francq2012strict} provide a  test for testing strict stationarity of a GARCH(1,1) process. Their results however do not lead to a handy tool for checking for the existence of an explosive period and determining its range.
There is also a large and important literature on testing for explosiveness and dating the period of instability in the {price or dividend processes} of a financial asset using a supreme unit root test for bubbles. See for example \cite{phillips2011explosive} for a left-tailed augmented Dickey-Fuller test (ADF) for the explosive behavior in the 1990s Nasdaq. \cite{hafner2018testing} considers such kind of bubble tests for crypto-currencies. \cite{harvey2018testing} investigate a bubble test with a smooth time varying volatility function. The underlying models focus usually on unit root or explosive auto-regressive (AR) processes to test the change of the AR(1) coefficient. Often the variance of the errors stays the same or vary smoothly after the explosion which means that the volatility increase is mostly driven by the increase of the AR parameter. In our model we choose a different approach to model an explosion of volatility: we describe the evolution of the data-generating process by a GARCH process and therefore link the source of a change in the volatility to a change of the parameters in the volatility recursion.

%It is worthwhile to study an analogue of such type of test in the volatility process, in particular to account for explosiveness in the GARCH process. {Thus we look at 'jumps' linked to the change of parameters in the volatility process and the change also affects the whole evolution of the future volatility.}

%They consider explosive processes that are symptoms of the presence of bubbles. 
It is worth noting that unlike a bubble test for an AR process, it is quite debatable to link a direct cause of the bursting behavior to the volatility process, see \cite{jurado2015measuring}. In contrary, volatility bursting can also be related to time varying risk aversion, sentiment, bubble or uncertainty. Nevertheless, we are trying to establish a theoretical rigorous framework of testing for the explosive interval using a GARCH model for the volatility process. It should be stressed that we focus on one aspect of the parameter and model uncertainty, namely changes in the parameters driving the volatility over time. We do not claim that our method can directly identify the cause of this behavior. In sum, we develop a change-point test for detecting  possible unstable behavior of a strictly stationary GARCH$(r,s)$ process. The null hypothesis is a GARCH process with globally constant parameters while the alternative is the existence of a period where the parameter values change to another (higher) value. This increase potentially leads to a period of explosive volatility.

Assuming that no information on the period and the change itself is available, we develop a test statistic based on supremes which searches over all possible sub-windows of the data. We prove asymptotic consistency and provide a limit distribution of our test statistic. It is important that the test is not of unit-root-type since hypothesis and alternative are still in the regime where the GARCH process is strictly stationary.
The theoretical contributions are uniform consistency statements of the QML estimators over an arbitrary observation period, a uniform weak Bahadur representation and corresponding uniform distributional limit results. For the proofs we carve out the essential analytical properties of the likelihood functions and use new concentration inequalities from \cite{zhangwu2017} leading to mild moment assumptions.

%Our theoretical contribution lies in three aspects: Firstly we have developed a uniform test scheme to detect changes of parameters of a GARCH process. Correspondingly we derive the uniform consistency of the quasi maximum likelihood function over any possible sub-period with changed parameter.  Secondly, the asymptotic distribution of our test statistics is derived based on a weak uniform Bahadur representation of our estimates. Thirdly, both the consistency and distribution theorem are built on carefully checking the analytical property of the likelihood and developing concentration inequalities for the empirical process involved.

%%%%%%%
We introduce some notations we use throughout the paper. For $q > 0$ and vector $v = (v_1, \ldots, v_d)^\top \in \IR^d,$ let $|v |_1 := \sum_{i=1}^d |v_i|$. For matrices $A \in \IR^{d\times d}$, we similarly  use $|A|_{1} := sum_{i,j=1}^d |A_{i,j}|$. We denote by $|A|_2 = \max_{|v| = 1} |A v|_2$ the spectral norm of $A$. We use $Z_n\dto Z$ and $Z_n\pto Z$ to denote convergence in distribution and convergence in probability for random variables $Z_n,Z$.
For some sequence $(y_j)_{j\in\IN}$ of real numbers and some sequence $(\chi_j)_{j\in\IN}$ of nonnegative real numbers, define the weighted seminorm
\[
    |y|_{\chi,q} := \Big(\sum_{j=1}^{\infty}\chi_j |y_j|^q\Big)^{1/q}.
\]
For some sequences $(a_n)$ and $(b_n)$ of positive numbers, write $a_n=O(b_n)$ or $a_n = o(b_n)$ if there exists a positive constant $C$ such that $a_n/b_n\leq C$ or $a_n/b_n\rightarrow 0$ respectively.
For two sequences of random variables $(X_n)$ and $(Y_n),$ write $X_n=o_{p}( Y_n)$ (resp. $X_n = O_p(Y_n)$) if $X_n/Y_n\rightarrow 0$ in probability ($X_n/Y_n$ is bounded in probability).
For some random variable $Z$, define $\|Z\|_q := (\IE|Z|^{q})^{1/q}$. If $\|\cdot\|_q$ is applied to a matrix, this is meant by a component-wise operation. For the i.i.d. random variables $\zeta_i$, $i\in\IZ$ used in the model definition \reff{eq:garchrecursion}, let $\sF_i := (\zeta_i,\zeta_{i-1},...)$. With some abuse of notation, we refer to $\sF_i$ also as the $\sigma$ algebra generated by the entries of $\sF_i$. For $X_i = h(\sF_i)$, $i\in\IZ$ with some measurable function $h$, we define the functional dependence measure (cf. \cite{wushao2004}),
\[
    \delta_q(k) := \|X_i - X_i^{*}\|_q,
\]
where $X_i^{*} = h(\sF_i^{*})$ and $\sF_i^{*} := (\zeta_i,...,\zeta_{1},\zeta_0^{*},\zeta_{-1},\zeta_{-2},...)$ with $\zeta_0^{*}$ being an independent copy of $\zeta_0$.

Our text is organized as follows.
Section \ref{sec:model} provides the results for the important GARCH(1,1) model and the corresponding test procedure. Section \ref{garch1} is concerning the estimation and theoretical results in a general GARCH($r,s$) model. In particular, Section \ref{sec:garch} introduces the framework of the QMLE and the consistency of the QMLE, and Section \ref{sec:th} presents the theoretical foundations of our uniform test and discusses the estimation of the covariance matrix of the QMLE appearing in the test statistic. 
In Section \ref{sec:sim} we analyze the size and the power of our test in simulations, while Section \ref{sec:app} discusses the behavior of the test in examples from practice. The technical proofs are delegated to the Appendix.

\section{A supreme explosiveness test for GARCH(1,1)}\label{sec:model}
In this section, we introduce our model by starting with a simple testing framework for the GARCH(1,1) model. Then we will provide a rigorous theoretical treatment by starting with a more general GARCH($r,s$) model in the following section.
We consider first of all the baseline GARCH(1,1) model over the whole sample period with possibly time-varying parameters, 
\begin{eqnarray}
    X_i &=& \zeta_i\sigma_i,\nonumber\\
    \sigma_i^2 &=& \alpha_0(i) + \alpha_1(i) X_{i-1}^2 + \beta_1(i) \sigma_{i-1}^2,\quad i\in\IZ,\label{eq:garch11}
\end{eqnarray}
where $\zeta_i$ is an i.i.d. sequence of random variables with $\IE \zeta_1 = 0$, $\IE \zeta_1^2 = 1$, and $\alpha_0(i), \alpha_1(i), \beta_1(i) > 0$ are the underlying parameters at each time point.
We collect data of this model at time points $1,\dots, n$.

We summarize the parameters into $\theta(i) = (\alpha_0(i),\alpha_1(i),\beta_1(i))'$. In the case that the parameters are constant, i.e. $\theta(i) \equiv \theta = (\alpha_0',\alpha_1',\beta_1')'$, the top Lyapunov exponent associated with this model according to \cite{gammaconditionpaper} is 

\[
    \gamma(\theta) = \IE \log(\alpha_1 \zeta_1^2 + \beta_1).
\]
%In particular it is shown that for example in \cite{francq2012strict} if $\gamma(\theta)<0$, then the conditional volatility $\sigma_i$ converges almost surely to $\sigma_{i,\infty}$ as $i \to \infty$, with  $\sigma_{i,\infty} = \lim_{n\to \infty}   \alpha_0^{*}\{1+ \sum^{n-1}_{k=1} \log(\alpha^*_1 \zeta_{t-1}^2 + \beta^*_1)\cdots \log(\alpha^*_1 \zeta_{t-k}^2 + \beta^*_1)$.
It is worth noting that $\gamma(\theta)<0$ allows (for instance) the IGARCH case, i.e. $\alpha_1+\beta_1 = 1$. The aim of this paper is to construct tests for a period of changed parameters in the GARCH model, which allows for the  explosiveness of the variances of the process. Namely, we would like to test whether there exists a period $\{n_1,...,n_2\}$ (with $1 < n_1 < n_2 < n$), where the parameter values
 in (\ref{eq:garch11}) change there values compared with $\{1,...,n\}$. The task breaks into two parts: First, checking for the existence of a change, for which a uniform test is needed. Second, one has to identify the period of the change and to estimate the corresponding parameters. 
%Furthermore, we certainly would like to make inference on our estimated parameters.

For our studies, let
\[
    \Theta = \{\theta = (\alpha_0,\alpha_1,\beta_1) \in\IR^3: \gamma(\theta)<0, \alpha_0,\alpha_1,\beta_1 > 0\}
\]
be the parameter space which contains all possible configurations of $\theta = (\alpha_0,\alpha_1,\beta_1)$. 

Let $\theta^*(i) = (\alpha_0^{*}(i),\alpha_1^{*}(i),\beta_1^{*}(i))'$ denote the true parameter in the baseline model, which possibly has a period of change in $\{\lfloor n\tau_1^{*}\rfloor+1,..., \lfloor n\tau_2^{*} \rfloor\}$ (where $\tau_1^{*},\tau_2^{*} \in [0,1]$, $\tau_1^{*} < \tau_2^{*}$) with a magnitude $\Delta^* = [\delta_1^*, \delta_2^*,\delta_3^*]'$. Namely,
\[
    \theta^{*}(i) = \begin{cases}
    \theta^{*}, & i \le \lfloor n\tau_1^{*}\rfloor,\\
    \theta^{*} + \Delta^{*}, & \lfloor n\tau_1^{*}\rfloor+1 \le i \le \lfloor n\tau_2^{*} \rfloor,\\
    \theta^{*}, & i > \lfloor n\tau_2^{*}\rfloor.
\end{cases}
\]
 Here $\lfloor x \rfloor$ denotes the flooring operator, i.e. the largest integer smaller or equal to $x$. An interesting question is to test whether the process is stable, i.e. $\alpha_1^*(i)+ \beta_1^*(i)<1$ for all time points $i = 1,...,n$ versus the hypothesis that there exists a period explosiveness in which $\alpha^*_1(i)+ \beta^*_1(i)>1$ for some $i$. $\alpha_1^*(i)+ \beta_1^*(i)$ is referred to as the persistency parameter in our setting. Graphically, this corresponds to the question if there exists regions where the process leaves the variance-stationary regime (i.e. the variance explodes).

We formulate our hypothesis therefore in the following way,  with $\alpha_1^*+ \beta_1^* = c$ ($c\le 1$),
\[
        H_0: \delta^*_2+\delta_3^* = 0 \quad\text{v.s.} \quad H_1: \delta_2^{*}+\delta^*_3> 0.
\]
To construct a test, we first derive estimators for the parameters. If the period described by $\tau_1^{*},\tau_2^{*}$ is known, we can use a standard 
quasi maximum likelihood estimation approach.
It is not hard to see from (\ref{eq:garch11}), that in the case of constant parameters $\theta(i) \equiv \theta$,
\[
    \sigma_i^2 = \alpha_0/(1-\beta_1)+ \alpha_1 \sum^{\infty}_{k=1}\beta_1^k X^2_{i-1-k}\quad \text{a.s.}
\]
The truncated version which can be calculated from a sample is
\[
    \sigma_i^{2c} =  \alpha_0/(1-\beta_1)+ \alpha_1 \sum^{i-2}_{k=1}\beta_1^k X^2_{i-1-k}.
\]
The quasi likelihood approach is to use the negative log likelihood function
\[
    L_{n,\tau_1,\tau_2}^c(\theta) := \frac{1}{n}\sum_{i=\lfloor n\tau_1\rfloor+1}^{\lfloor n\tau_2\rfloor}\ell(X_i^2,Y_i^c,\theta),
\]
where $Y_i^c := (X_{i-1}^2,...,X_{1}^2,0,0,...)$ and
\[
    \ell(X_i^2,Y_i^c,\theta) := \frac{1}{2}\Big( \frac{X_i^2}{\sigma_i^{2c}} + \log \sigma_i^{2c}\Big).
\]
The estimated parameter with observations during any given period $\{\lfloor n\tau_1\rfloor+1,..., \lfloor n\tau_2 \rfloor\}$ is defined to be $\hat \theta_{n,\tau_1,\tau_2} = \argmin_{\theta \in \Theta} L_{n,\tau_1,\tau_2}^c(\theta).$
It can be shown that under regularity conditions, $\hat \theta_{n,\tau_1,\tau_2}$ is asymptotically normal with covariance matrix
\begin{equation}
    \Sigma = V(\theta^{*})^{-1}I(\theta^{*})V(\theta^{*})^{-1}\label{eq:sigma_representation},
\end{equation}
where
\[
    V(\theta) := \IE[\nabla_{\theta}^2 \ell(X_i^2,Y_i,\theta)],
\]
and
\[
    I(\theta) := \IE[\nabla_{\theta}\ell(X_i^2,Y_i,\theta)\cdot \nabla_{\theta}\ell(X_i^2,Y_i,\theta)'].
\]

Now we turn to our first step, i.e. a uniform test of the existence of the period of change.
For given $\tau_1, \tau_2$, the test statistic associated with our hypothesis $H_0$ of interest is
\[
    T_{\tau_1,\tau_2} = (\tau_2 - \tau_1)^{\chi}(H'\bar \Sigma_{n,\tau_1,\tau_2} H)^{-1/2}\{\hat \alpha_{1,n,\tau_1,\tau_2}+ \hat \beta_{1,n,\tau_1,\tau_2}- c\},
\]
where $\chi$ is a scaling factor chosen for numerical stability, $\hat \alpha_{1,n,\tau_1,\tau_2}$, $\hat \beta_{1,n,\tau_1,\tau_2}$ are the second and third component of $\hat \theta_{n,\tau_1,\tau_2}$, and $\bar \Sigma_{n,\tau_1,\tau_2}$ is an estimator of $\Sigma$ using observations outside $\{\lfloor n\tau_1 \rfloor,..., \lfloor n\tau_2 \rfloor\}$.

For instance, we can set $\chi = 0.5$,  $\bar\Sigma_{n,\tau_1,\tau_2} $ to be the standard covariance matrix estimator obtained by replacing $V,I$ with their empirical counterparts (see \reff{eq:estimator_sigma} below) with observations outside $\{\lfloor n\tau_1 \rfloor,...,\lfloor n\tau_2 \rfloor\}$ and $H = (0,1,1)^{\top}$.

The feasible search set is defined to be $R_{\kappa,\kappa'} := \{(\tau_1,\tau_2) \in [0,1]^2: \tau_1 < \tau_2, 1-\kappa' \ge \tau_2 - \tau_1 \ge \kappa\}$ (with some $\kappa,\kappa' > 0$, for instance $\kappa = \kappa' = 0.1$), ensuring proper estimation of $\Sigma$ due to $1-\kappa' \ge \tau_2 - \tau_1$ and a change detection based on enough samples due to $\tau_2 - \tau_1 \ge \kappa.$
%The uniformity test is thus taken on the set  $R_{\kappa,\kappa'}$ to be any combination of $\tau_1, \tau_2$ with 
%\kappa<|\tau_1-\tau_2|\leq 1-\kappa'$.
The supreme of 
$\sqrt{n}T_{\tau_1, \tau_2}$ converges asymptotically to the supreme of a Gaussian process, namely $\{\frac{B(\tau_2) - B(\tau_1)}{(\tau_2 - \tau_1)^{1-\chi}}\}$, where $B(.)$ is a $1-$dimensional Brownian motion.  We show this formally in Theorem \ref{theorem:gaussian} in Section \ref{sec:th}.

Empirically, we cannot exhaust all the values $(\tau_1, \tau_2) \in R_{\kappa,\kappa'}$. 
%for the ease of implementation and derivation we define our feasible search set to be ${\kappa,\kappa'} := \{(\tau_1,\tau_2) \in [0,1]^2: \tau_1 < \tau_2, 1-\kappa' \ge \tau_2 - \tau_1 \ge \kappa\}$, where $\kappa, \kappa'$ are the bound on the distance between $(0,1)$.
We therefore need to restrict the calculation of the supreme to a set of grid points as an approximation to our supreme test statistics.

We summarize our uniform test as follows.
\begin{enumerate}
    \item[Step 0] Choose some $L > 0$ (number of grids associated with detection accuracy). The corresponding grid points are $\sG = \{\frac{j}{L}:j=0,...,L\}$ on the time line.
    \item[Step 1]  Let $H = (0,1,1)'$.
    %Note that alternatively we can also select a hypothesis $H \in \IR^3$ and a different value $\theta^{*} \in \Theta$. Choose some significance level $\delta \in (0,1)$, e.g. $\delta = 0.95$.
    \item[Step 2]  Choose values for $\kappa, \kappa'\in (0,1)$, and $\chi \in (0,1)$.
    \item[Step 3] For each given interval $(\tau_1,\tau_2) \in R_{\kappa,\kappa'}\cap \sG^2$, determine the associated QLME $\hat \theta_{n,\tau_1,\tau_2}$ defined in \reff{eq:estimator_theta} and calculate $\bar \Sigma_{n,\tau_1,\tau_2}$ as in \reff{eq:estimator_sigma}. Then we can calculate the test statistics as
    \[
        \hat B_n(\tau_1,\tau_2) := \sqrt{n}(\tau_2 - \tau_1)^{\chi} \{H'\bar\Sigma_{n,\tau_1,\tau_2}H\}^{-1/2}\big\{H' \hat \theta_{n,\tau_1,\tau_2} - H' \theta^{*}\big\}.
    \]
    (cf. \reff{eq:finalteststatistic_garch11}).
    \item[Step 4] For the critical value of this test, we can approximate the quantile of the test statistics via simulation under the null hypothesis $H_0$: For large $N$  (e.g. $N = 10000$), generate i.i.d. $\varepsilon_{i,k} \sim N(0,1)$, $i = 1,...,n$ and calculate
    \[
        \hat \mu_{n,k} := \sup_{(\tau_1,\tau_2) \in R_{\kappa,\kappa'} \cap \sG^2}\frac{1}{\sqrt{n}}\frac{\sum_{i=\lfloor n\tau_1\rfloor+1}^{\lfloor n \tau_2 \rfloor}\varepsilon_{i,k}}{(\tau_2 - \tau_1)^{1-\chi}}.
    \]
    We define $\hat q_{W,\delta} := \hat \mu_{n,[\lfloor N \cdot \delta\rfloor]}$, where $\hat \mu_{n,[1]},...,\hat \mu_{n,[N]}$ are the order statistics of $\hat \mu_{n,1},...,\hat \mu_{n,N}$.
     Figure \ref{figure:al} shows how one calculates the supreme test statistics over different windows associated with the grid points.
    
    \item[Step 5] We can now make a test decision based on critical values from the previous steps. If 
    \begin{equation}
        \sup_{(\tau_1,\tau_2) \in R_{\kappa,\kappa'} \cap \sG^2}\hat B_n(\tau_1,\tau_2) > \hat q_{W,\delta},\label{eq:testforgarch11}
    \end{equation}
    there is a significant shock in the parameter values. In this case, one can estimate the true shock period as $[\tau_1^{*}, \tau_2^{*}]$ 
    \[
    (\hat \tau_{1,n},\hat \tau_{2,n}) \in \argmax_{(\tau_1,\tau_2)\in R_{\kappa,\kappa'}}(\tau_2-\tau_1)^{\chi}\bar\Sigma_{n,\tau_1,\tau_2,H}^{-1/2}\big\{H' \hat \theta_{n,\tau_1,\tau_2} - H' \theta^{*}\big\}.
\]
   If instead \reff{eq:testforgarch11} does not hold, we conclude that there is no evidence for a period of parameter change.
    \item[Step 6]
   In case of the significance of our uniform test in Step 5, we re-estimate the parameter $\theta_{n,\hat{\tau}_1, \hat{\tau}_2}$, and produce the confidence interval based on Theorem \ref{theorem:gaussian}.
\end{enumerate}
\begin{figure}[h] 
\caption{Plot of the windows where the supreme is calculated. }
\centering
\label{figure:al}
\includegraphics[width=0.5\textwidth]{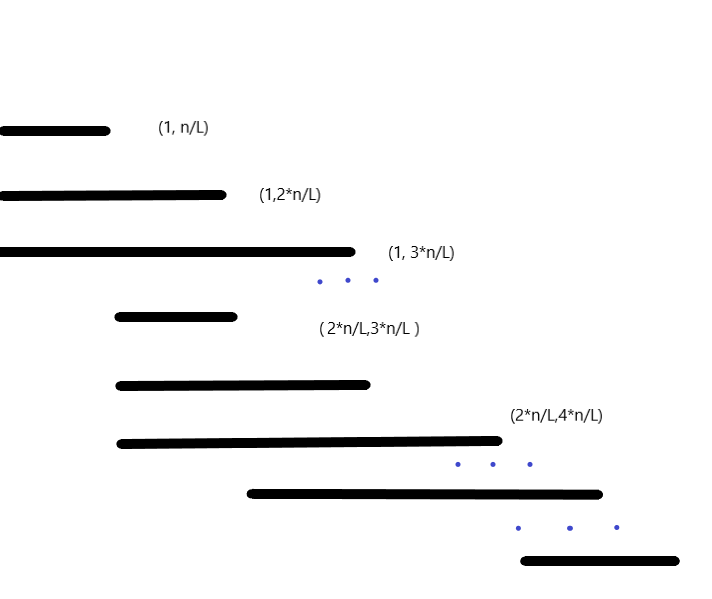}
\end{figure}

We name our test as a GSRWW test  (stands for GARCH supreme and the author abbreviations). The GSRWW test's limit distribution is 
$\mbox{sup}_{\tau_1,\tau_2 \in R_{\kappa,\kappa'}}\{\frac{B(\tau_2) - B(\tau_1)}{(\tau_2 - \tau_1)^{1-\chi}}\}$ (with $B(.)$ a standard Brownian motion).

\section{A supreme test for $GARCH(r,s)$  } \label{garch1}

In this section, we formulate the test and provide necessary theoretical results in a general GARCH($r,s$) model. 
For $r,s\in\IN$, $\theta(i) = (\alpha_0(i),\alpha_1(i),...,\alpha_r(i),\beta_1(i),...,\beta_s(i))'$, we consider the GARCH($r,s$) model
\begin{eqnarray}
    X_i^2 &=& \zeta_i^2 \sigma_i^2,\nonumber\\
    \sigma_i^2 &=& \alpha_0(i) + \sum_{j=1}^{r}\alpha_j(i) X_{i-j}^2 + \sum_{k=1}^{s}\beta_k(i) \sigma_{i-k}^2.\label{eq:garchrecursion}
\end{eqnarray}
Here, $\zeta_i$ are i.i.d. innovations with $\IE \zeta_1 = 0$ and $\IE \zeta_1^2 = 1$.
We first analyze the model in the case that the parameters are constant, i.e. $\theta(i) \equiv \theta = (\alpha_0,\alpha_1,...,\alpha_r,\beta_1,...,\beta_s)'$.
We first present our set of assumptions ensuring the existence of a unique stationary solution to our model in (\ref{eq:garchrecursion}) . Define $f(\theta) = (\alpha_1,\ldots,\alpha_r,\beta_1,\ldots,\beta_s)'$ and let $e_j = (0,\ldots,0,1,0,\ldots,0)' \in \IR^{r+s}$ be the unit column vector with $j$th element being 1, $1 \le j \le r+s$. Define the $(r+s) \times (r+s)$-matrix
\[
    A_i(\theta) = (f(\theta)\zeta_i^2, e_1,\ldots,e_{r-1},f(\theta),e_{r+1},\ldots,e_{r+s-1})'.
\]
Recall that $|A|_2$ is the spectral norm of a quadratic matrix $A$. Define the top Lyapunov exponent of $A_i(\theta)$,
\[
    \gamma(\theta) := \lim_{i\to\infty}\frac{1}{i}\log |A_{i}(\theta)A_{i-1}(\theta)\dots A_{1}(\theta)|_2.
\]
which exists if $\IE|\zeta_0^2|^{a} < \infty$ for some $a > 0$ (cf. \cite{garch2004}).

\begin{assumption}\label{ass1} Suppose that
    \begin{enumerate}
        \item[(A1)] $\zeta_0^2$ has a non-degenerate distribution with $\IE \zeta_0^2 = 1$.
        \item[(A2)] Let $\alpha_{min} > 0$, and 
        \begin{equation}
            \tilde \Theta = \{\theta \in \IR_{\ge 0}^{r+s+1}: \alpha_0 \ge \alpha_{min}, \gamma(\theta) < 0\text{ a.s.}, \sum_{j=1}^{s}\beta_j < 1\}.\label{eq:parameterspace}
        \end{equation}
        Let $\Theta \subset \tilde \Theta$ be compact. Assume that $\theta^{*} \in int(\Theta)$.
        \item[(A3)] Let $\sA_{\theta}(z) := \sum_{i=1}^{r}\alpha_i z^i$, $\sB_{\theta}(z) := 1-\sum_{j=1}^{s}\beta_j z^j$. If $s > 0$, $\sA_{\theta^{*}}(z)$ and $\sB_{\theta^{*}}(z)$ have no common root, $\sA_{\theta^{*}}(1)\not= 0$ and $\alpha^{*}_r + \beta^{*}_s \not= 0$.
    \end{enumerate}
\end{assumption}

{
Condition (A2) $\gamma(\theta) < 0$ ganrantees the strict stationarity of the GARCH process. 
Note that this includes parameter values corresponding to IGARCH or explosive GARCH $\sum_i \alpha_i + \sum_j \beta_j \geq 1$.
}

Most of the following proposition is in the spirit of the introductory comments in \cite{garch2004}. However, to obtain appropriate statements for the functional dependence measure, we have to modify some arguments in the proof.

\begin{proposition}[Existence of the GARCH model]\label{prop:existencegarch}
Let Assumption \ref{ass1} hold. Then:
\begin{enumerate}
    \item [(i)] \reff{eq:garchrecursion} has a unique stationary solution $X_i^2 = H(\sF_i)$, $i \in\IZ$.
    \item[(ii)] There exists $q > 0$ with $\|X_0^2\|_{q} \le D$ and $\delta_q^{X^2}(k) = O(c^k)$ for some $0 < c < 1$.
    \item[(iii)] $\lambda_{max}(B(\theta)) < 1$, where
    \[
        B(\theta) = \begin{pmatrix}
            \beta_1 & \beta_2 & \dots & \dots & \beta_s\\
            1 & 0 & \dots & \dots & 0\\
            0 & 1 & 0 & \dots & 0\\
            \vdots & \ddots & \ddots & \ddots & \vdots\\
            0 & \dots & 0 & 1 & 0
        \end{pmatrix}.
    \]
\end{enumerate}
\end{proposition}
\subsection{QMLE in $GARCH(r,s)$ and its consistency} \label{sec:garch}
In this subsection, we describe the QMLE and provide a theorem on its uniform consistency.
For estimation of the parameters $\theta \in \Theta$, we consider the following QML approach.
We denote by $Y_i^c := (X_{i-1}^2,X_{i-2}^2,...,X_1^2,0,0,...)$ the observed data until time $i-1$. For $0\le \tau_1< \tau_2\le 1$,
\[
    L_{n,\tau_1,\tau_2}^c(\theta) := \frac{1}{n}\sum_{i=\lfloor n\tau_1\rfloor+1}^{\lfloor n\tau_2\rfloor}\ell(X_i^2,Y_i^c,\theta),
\]
where
\[
    \ell(x,y,\theta) := \frac{1}{2}\Big( \frac{x}{\sigma(y,\theta)^2} + \log \sigma(y,\theta)^2\Big)
\]
and $\sigma(y,\theta)^2$ follows the recursion
\begin{equation}
    \sigma(y,\theta)^2 = \alpha_0 + \sum_{j=1}^{r}\alpha_j y_{j} + \sum_{k=1}^{s}\beta_k \sigma( (y_{k+1},y_{k+2},...),\theta)^2.\label{eq:estimation_sigma}
\end{equation}

{

}
{The analytic definition of the recursion of  $\sigma(y,\theta)^2$ is formulated in a forward way (using $y_1,y_2,...$ instead of $y_{-1}$,$y_{-2}$,...) because we plug in $y = Y_i^c$ which is formulated in a backward way, leading to the usual quasi-likelihood approach for GARCH models.} Note that $\sigma(Y_i^c,\theta)$ in \reff{eq:estimation_sigma} terminates after a finite number of steps due to zeros in $Y_i^c$.{ Morevoer, instead of using the truncated version $Y_i^c = (X_{i-1}^2,X_{i-2}^2,...,X_1^2,0,...,0)$ which corresponds to assuming that all initial values $X_0^2 = X_{-1}^2 = ... = 0$, one can also use different initial values like $X_0^2 = X_{-1}^2 = ... = \alpha_0$ or $X_0^{2} = X_{-1}^2 = ... = X_1^2$ as investigated in \cite{garch2004}}. For a discussion of different initial values, consider \cite{bougerol1992stationarity} (in the case of strict stationarity).

%Let $\sigma_t^2 = \alpha_0/(1-\sum^s_{k=1}\beta_k) +  \sum_{j=1}^{r}\alpha_j X^2_{t-j} + \sum_{j=1}^{r}\alpha_j \sum^{\infty}_{k=1}\sum^s_{j_1 = 1}\sum^s_{j_2 = 1}\cdots\sum^s_{j_k=1} \beta_{j_{1}} \beta_{j_{2}}\cdots\beta_{j_{k}} X^2_{t-i-j_1-\cdots-j_k}$.

With the defined likelihood function, for $0 \le \tau_1 < \tau_2 \le 1$, an estimator $\hat \theta_{n,\tau_1,\tau_2}$ of $\theta$ in the observation interval $i = \lfloor n\tau_1 \rfloor+1,...,\lfloor n \tau_2\rfloor$ is obtained by
\begin{equation}
    \hat \theta_{n,\tau_1,\tau_2} := \argmin_{\theta \in \Theta} L_{n,\tau_1,\tau_2}^c(\theta).\label{eq:estimator_theta}
\end{equation}

\begin{theorem}[Consistency of $\hat \theta_n$]\label{theorem:consistency}
    Let Assumption \ref{ass1} hold. Then for each $\kappa > 0$, 
    \[
        \sup_{0\le \tau_1 < \tau_2 \le 1,|\tau_1 - \tau_2| \ge \kappa}|\hat \theta_{n,\tau_1,\tau_2} - \theta^{*}|_1 \pto 0.
    \]
\end{theorem}

\subsection{Limiting distribution}\label{sec:th}
Given the consistency of our QMLE in a GARCH($r,s$) model, we provide a distribution theorem for $\hat \theta_{n,\tau_1,\tau_2}$ which allows us to obtain critical values for the uniform test defined in Section \ref{sec:model} and more general tests.
In our GSRWW test, we expect that over some observation period $\lfloor n\tau_1^{*}\rfloor+1,...,\lfloor n \tau_2^{*}\rfloor $, certain combinations of the parameters are large. For instance in the GARCH(1,1) model with $\theta = (\alpha_0,\alpha_1,\beta_1)$, one can observe (partly) explosive behavior even in the stationary case if $\alpha_1+\beta_1$ is large. To model the explosive behavior of the volatility process, we propose the following alternative 'shocked' $GARCH^{sh}(r,s)$ model where a change of parameter values happens at time $\lfloor n\tau_1^{*}\rfloor$. It pushes the parameters in a specific direction $H \in \IR^{(r+s+1) \times 1}$ which lasts until $\lfloor n\tau_2^{*}\rfloor$, where the parameter values go back to their initial states.

In the following, we assume that $\theta^{*}(i) = (\alpha_0^{*}(i),\alpha_1^{*}(i),...,\alpha_r^{*}(i),\beta_1^{*}(i),...,\beta_s^{*}(i))'$ denotes the true parameter in the baseline model, which possibly has a period of change in $\{\lfloor n\tau_1^{*}\rfloor+1,..., \lfloor n\tau_2^{*} \rfloor\}$ (where $\tau_1^{*},\tau_2^{*} \in [0,1]$, $\tau_1^{*} < \tau_2^{*}$). We suppose that the variation of $\theta^{*}$ over time reads
\[
    \theta^{*}(i) = \begin{cases}
    \theta^{*}, & i \le \lfloor n\tau_1^{*}\rfloor,\\
    \theta^{*} + H\Delta^{*}, & \lfloor n\tau_1^{*}\rfloor+1 \le i \le \lfloor n\tau_2^{*} \rfloor,\\
    \theta^{*}, & i > \lfloor n\tau_2^{*}\rfloor,
\end{cases}
\]
with some magnitude $\Delta^{*} \ge 0$ {such that $\theta^{*}+H\Delta^{*} \in \Theta$. If $\Delta^{*} = 0$, $\theta^{*}(i) \equiv \theta^{*}$ is constant over time and no change of parameter values happens; otherwise there is a change. We call the model \reff{eq:garchrecursion} with the above parameter configuration the shocked GARCH($r,s$) model, $GARCH^{sh}(r,s)$ for short.

The condition $\theta^{*}+H\Delta^{*} \in \Theta$ means that even in the alternative we assume that the observed process is strictly stationary. It should be noted that the space of allowed parameter configurations can be relaxed even further by sacrificing the estimation accuracy of the constant term $\alpha^*_0(i)$ (cf. \cite{francq2012strict}). }

We now give some important examples.

\begin{example}[$GARCH^{sh}(1,1)$]\label{example:garch11}
    Here, $\theta = (\alpha_0,\alpha_1,\beta_1)'$. Fix some $\bar\alpha_1$, for instance $\bar\alpha_1 = 1$, where we understand $\alpha^{*}_1 = \bar\alpha_1$ as a 'stable' parametrization of the process without a change. Assume that $X_i^2$ follows a $GARCH^{sh}(1,1)$ model with some $\Delta^{*} \ge 0$. The existence of a break period is related to  testing $\alpha_1^{*}(i) = \bar \alpha_1$ for all $i=1,...,n$ against $\alpha_1^{*}(i) > \bar\alpha_1$ for some $i$. Thus it corresponds to  $H = (0,1,0)'$,
    \[
        H_0: \Delta^{*} = 0 \quad\text{v.s.} \quad H_1: \Delta^{*} > 0.
    \]
\end{example}

\begin{example}[$GARCH^{sh}(r,s)$]\label{example:garchrs}
   For some fixed constants $\bar \alpha$, for instance $\bar\alpha = 1$, where we understand $\sum_{j=1}^{r}\alpha_j^{*} = \bar \alpha$ as a 'stable' parametrization. Assume that $X_i^2$ follows a $GARCH^{sh}(r,s)$ model with some $\Delta^{*} \ge 0$. The existence of an explosive period is then related to the question that whether $\sum_{j=1}^{r}\alpha_j^{*}(i) = \bar \alpha$ for all $i=1,...,n$ or $\sum_{j=1}^{r}\alpha_j^{*}(i) > \bar\alpha$ for some $i$. This is our test  with $H = (0,1,...,1,0,...,0)'$,
    \[
        H_0: \Delta^{*} = 0 \quad\text{v.s.} \quad H_1: \Delta^{*} > 0.
    \]
\end{example}

To formulate the hypotheses in a more general way, we propose the following 
\[
    H_0: H'\theta^{*} = \bar c \quad vs.\quad H_1: H'\theta^{*} > \bar c,
\]
where $H \in \IR^{(r+s+1)\times 1}$ is a vector, $\bar c \ge 0$. Motivated by the consistency result Theorem \ref{theorem:consistency}, we propose a test based on the supreme distance between the estimated targeting parameter value and the value under the null,
\[
    \hat B_{n,H}^{(\chi)} := \sqrt{n}\sup_{0 \le \tau_1 < \tau_2 \le 1, |\tau_1 - \tau_2|\ge \delta}(\tau_2 - \tau_1)^{\chi}\{H'\hat \theta_{n,\tau_1,\tau_2} - H'\theta^{*}\},
\]
where $\kappa > 0$ is some fixed parameter specifying the minimum length of the break period to be detected, and $\chi \in [0,1]$ is a scaling parameter which can be chosen arbitrarily, for instance $\chi = \frac{1}{2}$.

\begin{remark}[Consistency of the test statistic under the alternative]
If $H_1$ is true, there exists $\tau_1^{*} < \tau_2^{*}$ such that for $\lfloor n \tau_1^{*}\rfloor + 1 \le i \le \lfloor n\tau_2^{*}\rfloor$, $H'\theta^{*}(i) = H'\theta^{*} + H'H \Delta^{*} > H'\theta^{*}$, and  by applying Theorem \ref{theorem:consistency},
\[
    \hat B_{n,H}^{(\chi)} \pto \infty,
\]
which ensures that out test has asymptotic power 1.
\end{remark}
{We conjecture that the result can be extended even to nonstationary alternatives where  $\theta^{*}+H\Delta^{*} \not\in \Theta$ as long as $H'(1,0,...,0) = 0$ (\cite{francq2012strict} find out that one cannot expect $\hat \alpha_0$ to be consistently estimated in the nonstationary regime). Therefore in practice, we can fix $\alpha_0$ to be a constant and construct hypothesis in terms of other parameters.}

To obtain the critical values of our test, we need to derive quantiles for the test statistics $\hat B_{n,H}^{(\chi)}$, which can be inferred by its limit distribution. The asymptotic distribution of $\hat \theta_{n,\tau_1,\tau_2}$ is strongly connected to the Fisher information matrices
\[
    V(\theta) := \IE[\nabla_{\theta}^2 \ell(X_i^2,Y_i,\theta)],
\]
and
\[
    I(\theta) := \IE[\nabla_{\theta}\ell(X_i^2,Y_i,\theta)\cdot \nabla_{\theta}\ell(X_i^2,Y_i,\theta)'].
\]
We then present some properties regarding $V(\theta), I(\theta)$. From \cite{garch2004} (proof of Theorem 2.2, part (ii) therein), we directly obtain (ii),(iii) of the following Proposition:
\begin{proposition}(Properties of the variance covariance matrix of the parameter estimates)\label{prop:matrices}
    Let Assumption \ref{ass1} hold. Assume that $\mu_4 := \IE \zeta_0^4 < \infty$. Then:
    \begin{enumerate}
        \item[(i)] There exists $\iota > 0$ such that for all $\theta \in \Theta$ with $|\theta - \theta^{*}| < \iota$, $V(\theta)$ and $I(\theta)$ are finite.
        \item[(ii)] $I(\theta^{*}$) is nonsingular. It holds that $I(\theta^{*}) = \frac{\mu_4 - 1}{2}V(\theta^{*})$.
        \item[(iii)] $\IE \nabla_{\theta}\ell(X_i^2,Y_i,\theta^{*}) = 0$.
    \end{enumerate}
\end{proposition}

Next we provide results to quantify the difference $\hat \theta_{n,\tau_1,\tau_2} - \theta^{*}$ by a simple linear form uniformly in $\tau_1,\tau_2$. In the following, let $\kappa \in (0,1)$. We restrict ourselves to the case where $|\tau_1 - \tau_2| \ge \kappa$, i.e.
\[
    (\tau_1,\tau_2) \in R_{\kappa} := \{(\tau_1,\tau_2)\in[0,1]^2:\tau_1 < \tau_2, |\tau_1 - \tau_2| \ge \kappa\},
\]
where {$n\cdot \kappa$ then can be understood as the minimum size of a shock period which can be detected. In practice, $\kappa \in (0,1)$ can be chosen using a cross validation method so that the restriction on $R_{\kappa}$ does not affect the estimation accuracy.}

\begin{theorem}[Weak Bahadur Representation]\label{theorem:bahadur}
    Let Assumption \ref{ass1} hold. Assume that for some $a > 0$, $\IE |\zeta_0|^{4+a} < \infty$. Then for each $\kappa > 0$,
    \[
        \sup_{(\tau_1,\tau_2)\in R_{\kappa}}\big|\{\hat \theta_{n,\tau_1,\tau_2} - \theta^{*}\} + ((\tau_2 - \tau_1)V(\theta^{*}))^{-1}\cdot \nabla_{\theta}L_{n,\tau_1,\tau_2}(\theta^{*})\big| = O_{p}(\log(n)^3 n^{-1}).
    \]
\end{theorem}

With the linearization of our parameter estimation, we further obtain the limit distribution of $\hat \theta_{n,\tau_1,\tau_2}$ {uniformly in $(\tau_1,\tau_2) \in R_{\kappa}$} under $H_0$ by using Gaussian approximation results from \cite{zhouwu2011}. {It naturally implies the pointwise convergence results from \cite{garch2004} but is much stronger since it can be used as a starting point to apply theorems (such as the continuous mapping theorem) from the empirical process theory. Let $\ell^{\infty}(T)$ denote the space of bounded functions $f:T \to \IR$, cf. \cite{vandi}, Section 18, Example 18.5.} As a direct consequence of the uniform Bahardur representation, we can derive the distribution of a simple test statistics under the null.

\begin{theorem}[Asymptotic distribution of $\hat \theta_{n,\tau_1,\tau_2}$]\label{theorem:gaussian} Assume assumption \ref{ass1}. Suppose that there exists $a' > 0$ such that $\IE |\zeta_0|^{4+a'} < \infty$. Fix $\kappa > 0$. Then on $\ell^{\infty}(R_{\kappa})^{r+s+1}$,
\begin{eqnarray*}
    && \sqrt{n}\big\{\hat \theta_{n,\tau_1,\tau_2} - \theta^{*}\big\} \dto  \Sigma^{1/2}\{\frac{B(\tau_2) - B(\tau_1)}{\tau_2 - \tau_1}\},
\end{eqnarray*}
where $B(\cdot)$ is a standard $(r+s+1)$-dimensional Brownian motion and
\[
    \Sigma := V(\theta^{*})^{-1}I(\theta^{*})V(\theta^{*})^{-1} = \frac{\mu_4-1}{2}\cdot V(\theta^{*})^{-1},
\]
where $\mu_4 := \IE \zeta_0^4$.
\end{theorem}

With a more general formulation, we obtain the limit distribution of $\hat B_{n,H}^{(\chi)}$ with the continuous mapping theorem. We state a slightly more general result by letting $H \in \IR^{(r+s+1) \times d}$ which allows to detect for more than one deviation from a 'stable' state.

\begin{corollary}[Limit distribution of $\hat B_{n,H}^{(\chi)}$]\label{theorem:limitBn} Suppose that Assumption \ref{ass1} holds. Suppose that there exists $a' > 0$ such that $\IE |\zeta_0|^{4+a'} < \infty$. Fix $\kappa > 0$.
Let $H \in \IR^{(r+s+1)\times d}$ be a matrix with full rank. Let $\Sigma_H := H'\Sigma H$. Then
\begin{eqnarray*}
    \hat B_{n,H}^{(\chi)} &=&\sqrt{n}\sup_{(\tau_1,\tau_2) \in R_{\kappa}}(\tau_2 - \tau_1)^{\chi}\big\{H' \hat \theta_{n,\tau_1,\tau_2} - H' \theta^{*}\big\}\\
    &\dto& \Sigma_H^{1/2}\sup_{(\tau_1,\tau_2) \in R_{\kappa}}\{\frac{B(\tau_2) - B(\tau_1)}{(\tau_2 - \tau_1)^{1-\chi}}\},
\end{eqnarray*}
where $B(\cdot)$ is a standard $d$-dimensional Brownian motion.
\end{corollary}

%%%old
%The quantiles of $\sup_{(\tau_1,\tau_2) \in R_{\kappa}}\{\frac{B(\tau_2) - B(\tau_1)}{(\tau_2- \tau_1)^{1-\beta}}\}$ can be obtained by simulation.
%%%old

\subsection{Estimation of $\Sigma$}
In this subsection, we discuss how to estimate the variance covariance matrix of our QMLE. which is needed to use the proposed test without prior knowledge of the parameters. We have seen that
\begin{eqnarray}
    \Sigma &=& V(\theta^{*})^{-1}I(\theta^{*})V(\theta^{*})^{-1}\label{eq:sigma_representation}\\
    &=& \frac{\mu_4-1}{2}\cdot V(\theta^{*})^{-1}.\nonumber
\end{eqnarray}
Here we restrict ourselves to estimation of $\Sigma$ via the representation \reff{eq:sigma_representation} to avoid estimating $\mu_4$ separately. To get a test with high power it seems reasonable to estimate $\Sigma$ under the alternative. Recall that $R_{\kappa,\kappa'} := \{(\tau_1,\tau_2) \in [0,1]^2: \tau_1 < \tau_2, 1-\kappa' \ge \tau_2 - \tau_1 \ge \kappa\}$. Let
\[
    \bar L_{n,\tau_1,\tau_2}^c(\theta) := \frac{1}{n}\sum_{i \in \{1,...,n\}\backslash\{\lfloor n\tau_1 \rfloor+1,\lfloor n \tau_2 \rfloor\}}\ell(X_i^2,Y_i^c,\theta)
\]
and define the estimator of $\theta^{*}$ in the stationary regime,
\[
    \bar \theta_{n,\tau_1,\tau_2} := \argmin_{(\tau_1,\tau_2)\in  R_{\kappa,\kappa'}}\bar L_{n,\tau_1,\tau_2}^c(\theta).
\]
Now put
\begin{eqnarray*}
    \bar V_{n,\tau_1,\tau_2}(\theta) &:=& \frac{1}{1-(\tau_2 - \tau_1)}\nabla_{\theta}^2 \bar L_{n,\tau_1,\tau_2}^c(\theta),\\
    \bar I_{n,\tau_1,\tau_2}(\theta) &:=& \frac{1}{n(1-(\tau_2 - \tau_1))}\sum_{i \in \{1,...,n\}\backslash\{\lfloor n\tau_1 \rfloor+1,\lfloor n \tau_2 \rfloor\}}\nabla_{\theta}\ell(X_i^2,Y_i^c,\theta)\nabla_{\theta}\ell(X_i^2,Y_i^c,\theta)'.
\end{eqnarray*}

Then the following intermediate result holds:

\begin{proposition}[Estimation of $V,I$]\label{prop:convergence_matrices}
    Suppose that Assumption \ref{ass1} holds. Suppose that there exists $a' > 0$ such that $\IE |\zeta_0|^{4+a'} < \infty$. Fix $\kappa,\kappa' > 0$. Then:
    \begin{enumerate}
        \item[(i)] $\sup_{(\tau_1,\tau_2) \in R_{\kappa,\kappa'}}|\bar V_{n,\tau_1,\tau_2}(\bar \theta_{n,\tau_1,\tau_2}) -  V(\theta^{*})|_1 \pto 0$,
        \item[(ii)] If additionally $\IE |\zeta_0|^{8+a'} < \infty$, $\sup_{(\tau_1,\tau_2)\in  R_{\kappa,\kappa'}}|\bar I_n(\bar \theta_{n,\tau_1,\tau_2}) - I(\theta^{*})|\pto 0$.
    \end{enumerate}
\end{proposition}

We propose the estimate
\begin{equation}
    \bar \Sigma_{n,\tau_1,\tau_2} := \bar V_n(\bar \theta_{n,\tau_1,\tau_2})^{-1}\bar I_n(\bar \theta_{n,\tau_1,\tau_2}) \bar V_n(\bar \theta_{n,\tau_1,\tau_2})^{-1}\label{eq:estimator_sigma}
\end{equation}
for $\Sigma$ and $\bar \Sigma_{n,\tau_1,\tau_2,H} := H'\bar \Sigma_{n,\tau_1,\tau_2} H$. As a corollary of Theorem \reff{theorem:limitBn} and Proposition \reff{prop:convergence_matrices}, we obtain
\begin{eqnarray}
    \hat B_n &:=&\sqrt{n}\sup_{(\tau_1,\tau_2) \in R_{\kappa,\kappa'}}(\tau_2 - \tau_1)^{\chi}\bar \Sigma_{n,\tau_1,\tau_2,H}^{-1/2}\big\{H' \hat \theta_{n,\tau_1,\tau_2} - H' \theta^{*}\big\}\nonumber\\
    &\dto& \sup_{(\tau_1,\tau_2) \in R_{\kappa,\kappa'}}\{\frac{B(\tau_2) - B(\tau_1)}{(\tau_2 - \tau_1)^{1-\chi}}\} =: W\label{eq:finalteststatistic_garch11}
\end{eqnarray}
with $B(\cdot)$ a $d$-dimensional Brownian motion. The quantiles of the limit distribution of $W$ can be obtained via simulation.

%\subsection{Estimation of the break points $\tau_1^{*},\tau_2^{*}$}

If significance is detected,  $\tau_1^{*},\tau_2^{*}$ can be estimated by the choice
\[
    (\hat \tau_{1,n},\hat \tau_{2,n}) \in \argmax_{(\tau_1,\tau_2)\in R_{\kappa,\kappa'}}(\tau_2-\tau_1)^{\chi}\bar\Sigma_{n,\tau_1,\tau_2,H}^{-1/2}\big\{H' \hat \theta_{n,\tau_1,\tau_2} - H' \theta^{*}\big\},
\]
which is motivated by \reff{eq:finalteststatistic_garch11}. For the special case of one parameter change in a GARCH(1,1) model we will have the following simplified result.
%\subsection{Special case: GARCH(1,1) hypotheses testing}

\begin{example}[Example \reff{example:garch11} continued]
    From Corollary \reff{theorem:limitBn}, we obtain with under $H_0:\alpha_1^{*} = 1$:
    \begin{eqnarray*}
        && \sqrt{n}\sup_{(\tau_1,\tau_2) \in R_{\kappa,\kappa'}}(\tau_2 - \tau_1)^{\chi}(\bar \Sigma_{n,\tau_1,\tau_2})_{2,2}^{-1/2}\{\hat \alpha_{1,n,\tau_1,\tau_2} - 1\}\\
        &\dto& \sup_{(\tau_1,\tau_2) \in R_{\kappa,\kappa'}}\{\frac{B(\tau_2) - B(\tau_1)}{(\tau_2 - \tau_1)^{1-\chi}}\}
    \end{eqnarray*}
    with some 1-dimensional standard Brownian motion $B(\cdot)$.
\end{example}

\section{Simulation}\label{sec:sim}
In this section, we conduct a simulation study for evaluating the performance of our methodology. The algorithm is summarized in Section \ref{sec:model}.
We consider the shocked GARCH(1,1)-model with $\theta^{*}(i) = (\alpha_0^{*}(i), \alpha_1^{*}(i),\beta_1^{*}(i))'$,
\[
    \theta^{*}(i) = \begin{cases}
    \theta^{*}, & i \le \lfloor n\tau_1\rfloor,\\
     \tilde{\theta}^{*}, & \lfloor n\tau_1\rfloor+1 \le i \le \lfloor n\tau_2 \rfloor,\\
    \theta^{*}, & i > \lfloor n\tau_2\rfloor,
\end{cases}
\]
with $H_0: H'\tilde{\theta}^{*} = H \theta^*$, and $H_1:  H'\tilde{\theta}^{*} > H \theta^*$, where 
\begin{enumerate}
    \item[i)] $H = (0,1,0)'$ and $\alpha_0^{*} = 0.3$, $\alpha_1^{*} = 1.0$, $\beta_1^{*}=0.25$ or
    \item[ii)] $H = (0,1,1)'$ and $\alpha_0^{*} = 0.3$, $\alpha_1^{*} = 0.4$, $\beta_1^{*}=0.6$,
\end{enumerate}
and
\[
    X_i^2 = \zeta_i^2 \sigma_i^2, \quad\quad \sigma_i^2 = \alpha_0^{*}(i) + \alpha_1^{*}(i)X_{i-1}^2 + \beta_{1}^{*}(i)\sigma_{i-1}^2,
\]

We now check the behavior of the test under the null hypothesis $H_0$ of no change. We use the test proposed in Section 2 with $\chi = \frac{1}{2}$, $\kappa = \kappa' = 0.1$ and a grid approximation of $L = 30$.

%First we check the behaviour under $H_0$. We use $\chi = \frac{1}{2}$, $\kappa = \kappa' = 0.1$ and the test
%\[
%    \phi^{*} = \begin{cases}
 %       H_1, & \hat B_n^{GARCH(1,1)} > q_{W,\delta},\\
%        H_0, & \hat B_n^{GARCH(1,1)} \le q_{W,\delta},
%    \end{cases}
%\]
%with the notation from \reff{eq:finalteststatistic_garch11} and $q_{W,\delta}$ denoting the $\delta$-quantile of the limit distribution $W$ in (\ref{eq:finalteststatistic_garch11}). The critical value of our test $q_{W,\delta}$ can be obtained via simulation: For $k=0,...,N$, generate i.i.d. $\varepsilon_{i,k} \sim N(0,1)$, $i = 1,...,n$ and calculate
%\[
 %   \hat \mu_{n,k} :=  \sup_{(\tau_1,\tau_2) \in R_{\kappa,\kappa'} \cap \{\frac{j}{n}:j\in \{0,...,n\}\}^2}\frac{1}{\sqrt{n}}\frac{\sum_{i=\lfloor n\tau_1\rfloor+1}^{\lfloor n\tau_2\rfloor}\varepsilon_{i,k}}{(\tau_2 - \tau_1)^{1-\chi}}.
%\]
%If $\hat \mu_{n,[1]} \le ... \le \hat \mu_{n,[N]}$ denote the ordered values of $\hat \mu_{n,1},...,\hat \mu_{n,N}$, then $q_{W,\delta} \approx \hat \mu_{n,[\lfloor N\cdot \delta\rfloor]}$.\\
For $N = 1000$ replications and $n \in \{500,1000,2000\}$, $\delta \in \{0.90,0.95\}$, we obtain the quantiles (cf. Step 4 in the algorithm in Section \ref{sec:model}) $\hat q_{W,0.90} \approx 3.031$, $\hat q_{W,0.95} \approx 3.285$ and the results given in Table \ref{tab:data_hypothesis}. We find that the performance of the test is quite unaffected by the choice of $\chi$ and therefore do not present analysis for different values of $\chi$ here. We can see from the table that as the sample size increases, the coverage probabilities would approach to the nominal level for both $\delta = 0.90,0.95$. This illustrates a good performance of our test statistics.

\begin{table}[!ht]
\caption{Averaged acceptance rate under null hypothesis $H_0$. ($L = 30$).}
\centering
\label{tab:data_hypothesis}
\begin{tabular}{ l | c | c ||c|c}
\hline \hline
  &   \multicolumn{2}{c||}{i)} &  \multicolumn{2}{c}{ii)}  \\
 \hline
 $n$ / $\delta$ & $0.90$ & $0.95$ & $0.90$ & $0.95$  \\
\hline
500 &0.866 &0.907  &0.864 &0.903 \\
\hline
1000 & 0.877&0.906 & 0.884&0.910 \\
\hline
2000 & 0.896 & 0.918 & 0.859 & 0.913 \\
\hline\hline
\end{tabular}
\end{table}
To evaluate the test performance under the alternative, we consider $\delta = 0.95,0.90$ and the cases 
\[
    H'\tilde{\theta}^{*} - H \theta^* \in \{0.05,0.1,0.2\}
\]
with a shock period of $\tau^{*}_2 - \tau^{*}_1 \in \{0.1,0.2\}$, where we have chosen $\tau_1^{*} = 0.5$. The choice of $\tau_1^{*}$ does not have a big influence on the performance of the test, therefore we do not present simulation results for different $\tau_1^{*}$ here. The test results in different scenarios can be found in Table \ref{tab:data_alternative1} and \ref{tab:data_alternative2}. It can be seen that our test shows good power under the alternative hypothesis, which is robust against different choices of break sizes and time length of breaks.
We also find that as the sample size increases the power increases drastically.

\begin{table}[!ht]
\caption{Rejection rate of the test (test power) under the alternative $H_1$ with different $\Delta^{*}$ and $\tau_2^{*}-\tau_1^{*}$. ($L = 30$), $\delta = 95\%$}
\centering
\label{tab:data_alternative1}
\begin{tabular}{c| l || c | c | c || c | c | c }
\hline\hline
&& \multicolumn{3}{c||}{$\tau_2^{*} - \tau_1^{*} = 0.1$} &  \multicolumn{3}{c}{$\tau_2^{*}-\tau_1^{*}=0.2$}\\
\hline
& $n/\Delta^{*}$ & $\Delta^{*} = 0.05$ & $\Delta^{*} = 0.1$ & $\Delta^{*} = 0.2$ & $\Delta^{*} = 0.05$ & $\Delta^{*} = 0.1$ & $\Delta^{*} = 0.2$ \\
\hline
\multirow{3}{*}{i)}&500 &0.300 &0.305 &0.828&0.896&0.923 &0.942\\
\cline{2-8}
&1000 &0.571 &0.591 & 0.884&0.897 &0.908 & 0.920 \\
\cline{2-8}
&2000 & 0.801&0.813 & 0.902&1.000 &1.000 &1.000  \\
\hline\hline
\multirow{3}{*}{ii)}&500 &0.344&0.488 &0.831&0.901&0.940 &0.950\\
\cline{2-8}
&1000 &0.796 &0.801 & 0.865&0.965 &0.984 & 0.997 \\
\cline{2-8}
&2000 & 0.808&0.824 & 0.908&1.000 &1.000 &1.000  \\
\hline\hline
\end{tabular}
\end{table}

\begin{table}[!ht]
\caption{Rejection rate of the test (test power) under the alternative $H_1$ with different $\Delta^{*}$ and $\tau_2^{*}-\tau_1^{*}$. ($L = 30$), $\delta = 90\%$}
\centering
\label{tab:data_alternative2}
\begin{tabular}{c| l || c | c | c || c | c | c }
\hline\hline
&& \multicolumn{3}{c||}{$\tau_2^{*} - \tau_1^{*} = 0.1$} &  \multicolumn{3}{c}{$\tau_2^{*}-\tau_1^{*}=0.2$}\\
\hline
& $n/\Delta^{*}$ & $\Delta^{*} = 0.05$ & $\Delta^{*} = 0.1$ & $\Delta^{*} = 0.2$ & $\Delta^{*} = 0.05$ & $\Delta^{*} = 0.1$ & $\Delta^{*} = 0.2$ \\
\hline
\multirow{3}{*}{i)}&500 &0.376 &0.498 &0.855&0.917&0.936 &0.945\\
\cline{2-8}
&1000 &0.689 &0.711 & 0.900&0.903 &0.998 & 0.998 \\
\cline{2-8}
&2000 & 0.868&0.884& 0.922&1.000 &1.000 &1.000  \\
\hline\hline
\multirow{3}{*}{ii)}&500 &0.398 &0.711 &0.866&0.931&0.949 &0.960\\
\cline{2-8}
&1000 &0.803 &0.832 & 0.899&0.995 &0.999 & 0.999 \\
\cline{2-8}
&2000 & 0.945&0.968 & 0.996&1.000 &1.000 &1.000 \\
\hline\hline
\end{tabular}
\end{table}

\section{Application}\label{sec:app}
%We collect VIX index historical data from the website http://www.cboe.com/products/vix-index-volatility/vix-options-and-futures/vix-index/vix-historical-data .
%The data time span is from 2004 May 1st to 2008 May 9th. The historical volatility is calculated following two methods the Garman-Klass volatility estimator  as in \cite{garman1980estimation}  and the Yang-Zhang volatility estimator \cite{yang2000drift} using the high, low, open and close data.

We now apply our test to real data. We collect daily historical Apple stock prices (open, high, low, close and adjusted prices) from Jan 4th, 2000 to Nov 12th, 2018 extracted from the Yahoo Finance website, \url{http://finance.yahoo.com/quote/AAPL/history?p=AAPL}. The data has logged the prices of the Apple stock everyday and comprises of the open, close, low, high and the adjusted close prices of the stock for the span of 18 years. The goal of our analysis is to discover the existence of a period of unstable behavior for the underlying volatility process, and to see whether our test can help with forecasting. Figure \ref{figure:unplot} shows the plotted adjusted price, log returns and absolute log returns of the Apple stock price. From the plot we observe that the returns fluctuate around the zero. During the years 2000, 2008-2009, 2014 and 2015-16, there are  volatility peaks.
We divide the data into a sequence of  consecutive windows of $1000$ observations. The log returns are stationary in all the windows (suggested by the ADF test) and serial correlation is taken out by fitting an ARMA process in advance and the following analysis is done on the residuals.

\begin{figure}[h]
\caption{Plot of stock price (upper panel) and log returns of stock price of Apple (middle panel), absolute returns (lower panel).}
\centering
\label{figure:unplot}
\includegraphics[width=1\textwidth]{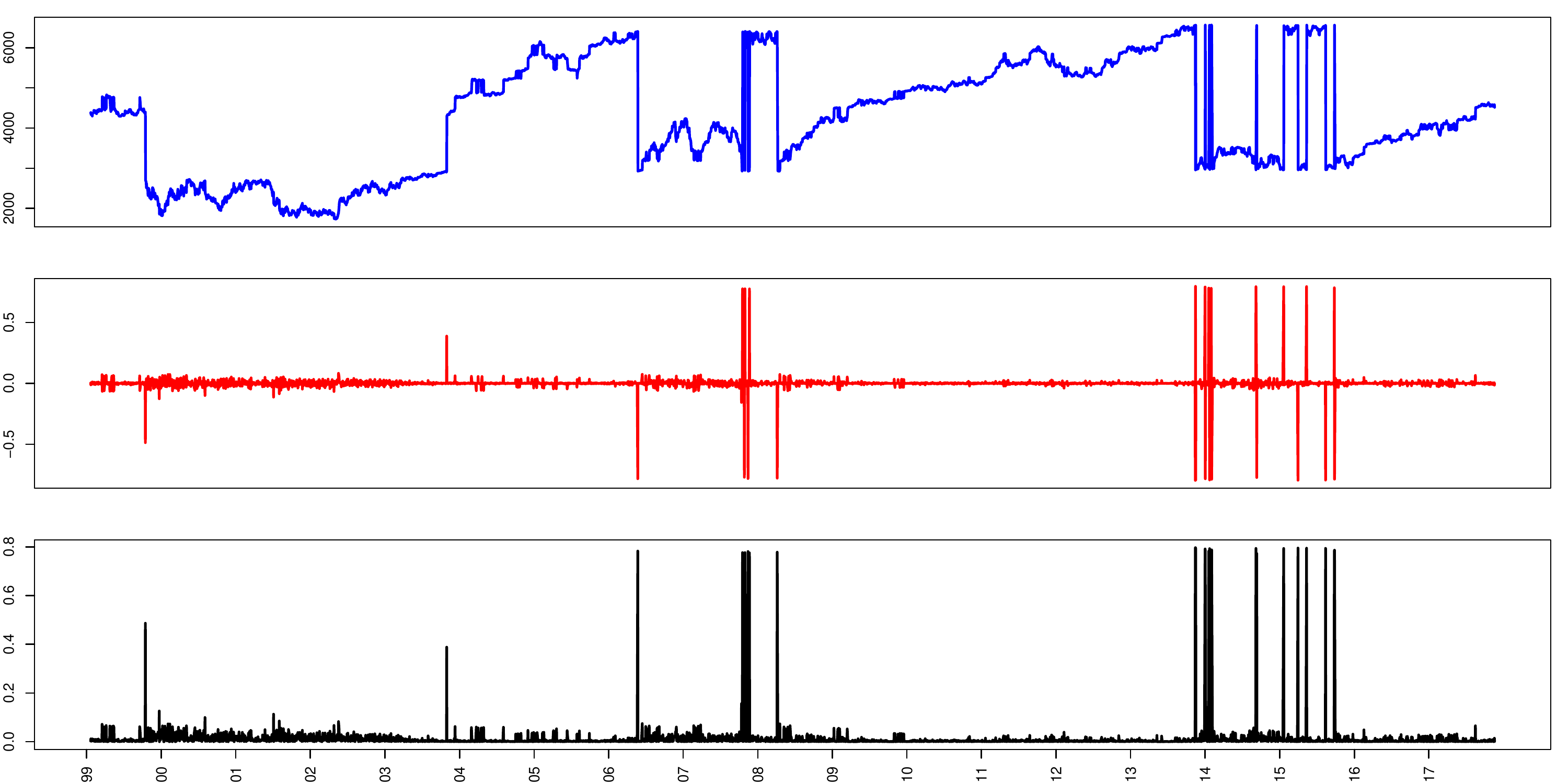}
\end{figure}
\begin{figure}[h]\label{figure:qq}
\caption{QQ plot and the histogram for the daily return of the stock prices.}
\centering
\label{figure:qq}
\includegraphics[width=1\textwidth]{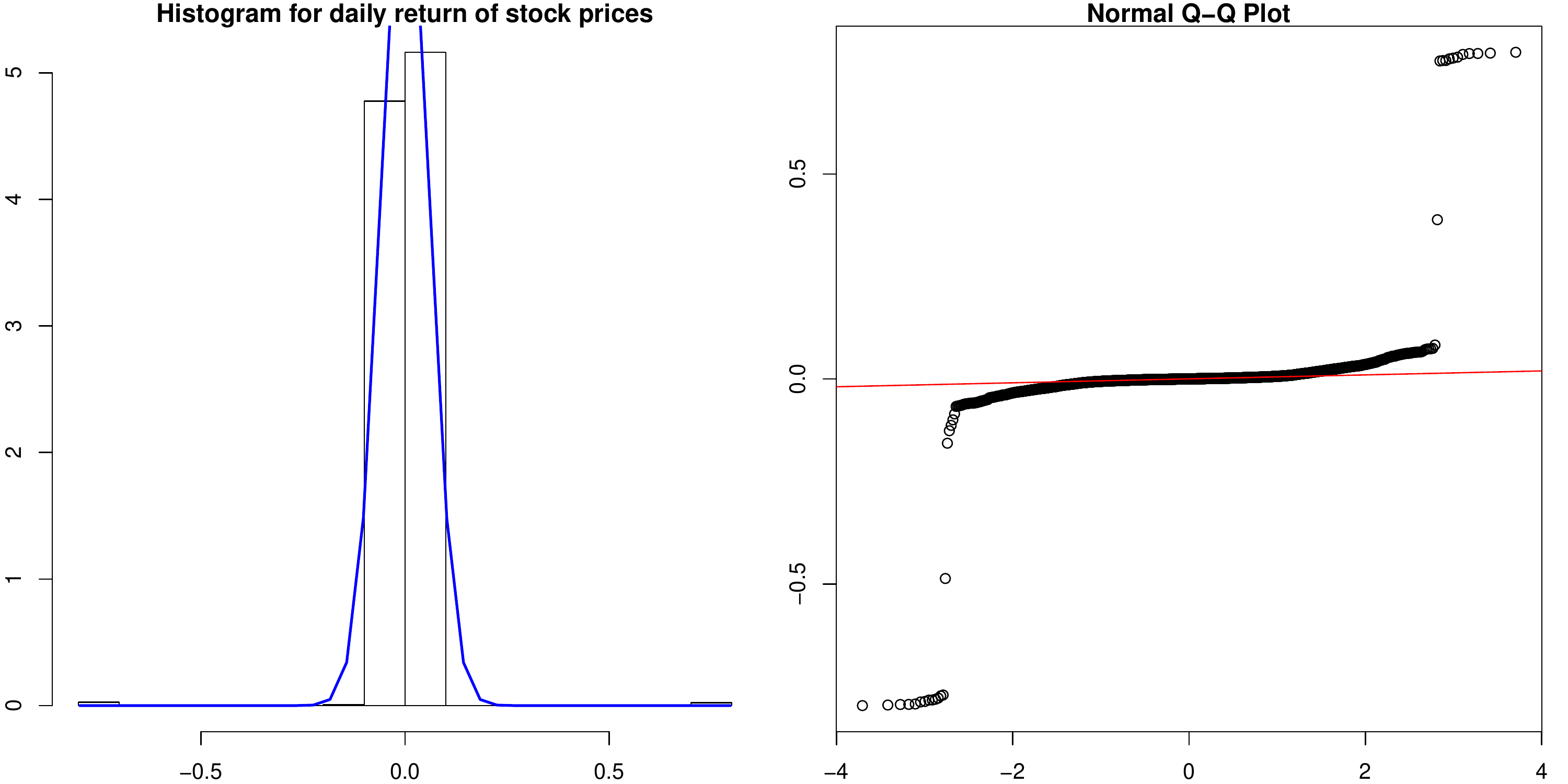}
\end{figure}

\begin{table}[]
\caption{Detected significant break periods for the apple returns, the corresponding persistency parameter($\hat{\alpha}_1+\hat{\beta}_1$) and the test statistics.  (***) means significant at both $0.95,0.90$, and (**) means significant at only $0.90$.} \label{break}
\begin{tabular}{l|l|l|l|l|c}
\hline\hline
  & $\hat{\tau}_1$   & $\hat{\tau}_2$    & in   & out& test statistics\\ \hline
1 & 2000/01/04 & 2000/06/27 & 1.01    & 0.98& 3.23(**) \\ \hline
2 & 2003/12/30 & 2004/10/15 & 1.27 & 0.98&3.76(***) \\ \hline
3 & 2008/06/12 & 2010/01/13 & 1.31 & 0.98&100(***) \\ \hline
4 & 2016/08/17 & 2018/09/11 & 1.33 & 0.88&224(***) \\ \hline\hline
\end{tabular}
\end{table}

From the histogram and Q-Q plot of the time series in Figure \ref{figure:qq} we observe a strong evidence of leptokurtic behavior. In Table \ref{break} we present the detected periods of the explosive behavior. The GSRWW test identifies major financial crises such as  the technological bubble in 2000 and the US subprime mortgage crisis between Dec. 2007 to June 2009. Furthermore, the test detects some short-lived instability such as the 2003 stock prices downturn.\\ 

%We will further conduct a forecast exercises on the returns %compared with a full window GARCH forecast without taking %into account the results of our test.
%We have found that the forecast would outperform the...

As a second analysis, we gather the Bitcoin price series ranging from July 23, 2015 to August 21, 2018 at a daily frequency and is presented throughout diagnostic tests. The data source is \url{https://coinmarketcap.com/currencies/bitcoin/historical-data/}. 
We show the returns and the absolute returns for the  Bitcoin in Figure \ref{figure:bitplot}. We can see that there are several high volatility periods. The volatility level is higher before 2013 followed by a stable period. Recently, the market volatility increased. The QQ plots and histograms are given in Figure \ref{figure:qqbit}, indicating the heavy-tailedness of the underlying distribution.
We present the test results with a window of $1000$ observations in Table \ref{bit:break}. Again the log returns are stationary in all the windows (by results of an ADF test) and serial correlation is taken out by fitting an ARMA process in advance. We apply our tests to the residuals obtained.
The GSRWW test indicates the presence of multiple market ‘euphoria’ episodes in the series. The  GSSWW identifies the most significant high volatility period including the period covering the June 2016 crash, the crashes during summer 2017, the fear of market regulation in October 2017, and the massive crash that commenced in December 2017. Bitcoin are not controlled by any government, but speculators can use the test results as indicators of the market sentiment.

\begin{figure}[h]
\caption{Plot of Bitcoin price (upper panel) and log returns of Bitcoin  (middle panel), absolute returns (lower panel).}
\centering
\label{figure:bitplot}
\includegraphics[width=1\textwidth]{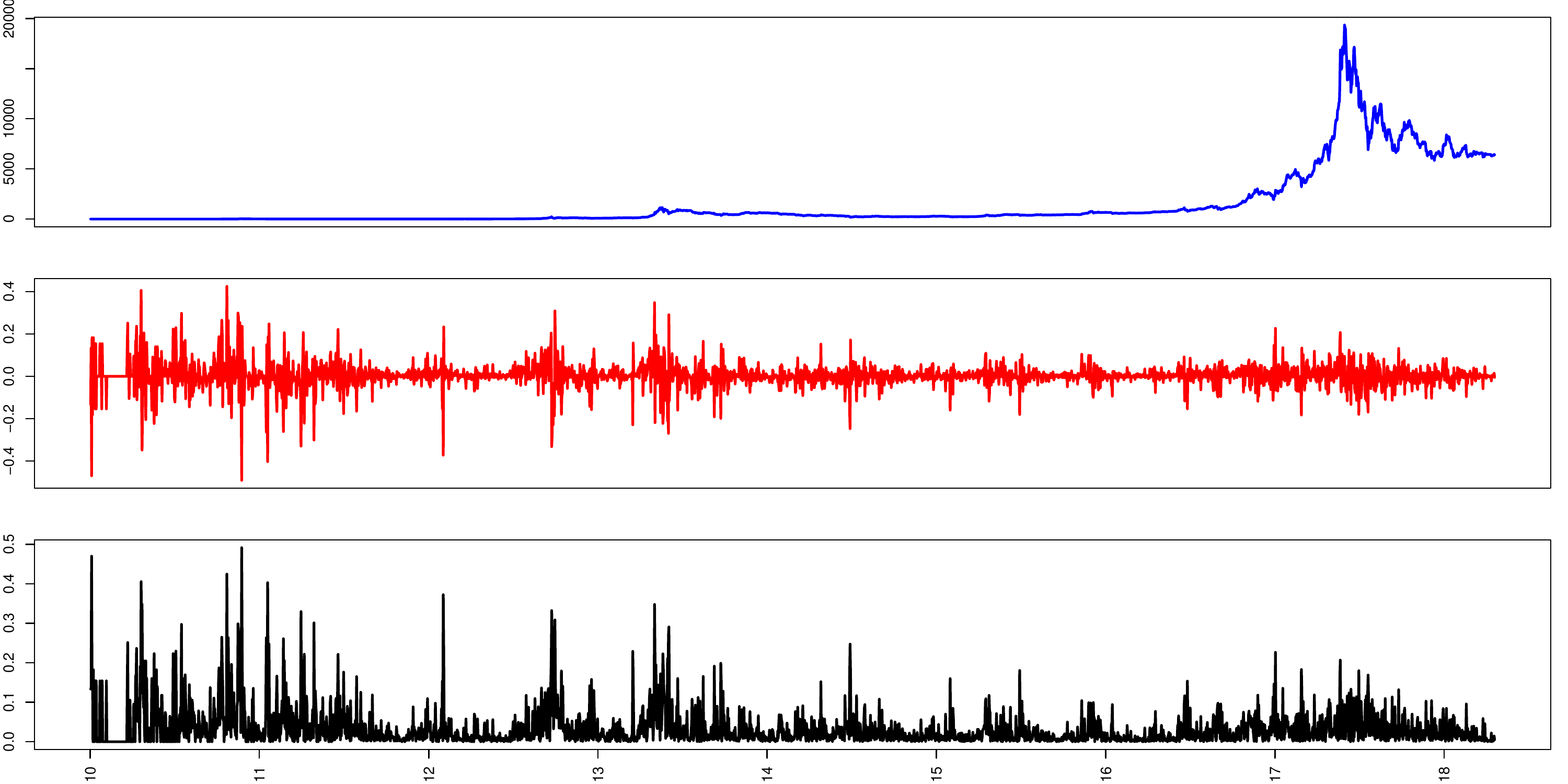}
\end{figure}

\begin{figure}[h]
\caption{QQ plot and the histogram for the Bitcoin returns}
\centering
\label{figure:qqbit}
\includegraphics[width=1\textwidth]{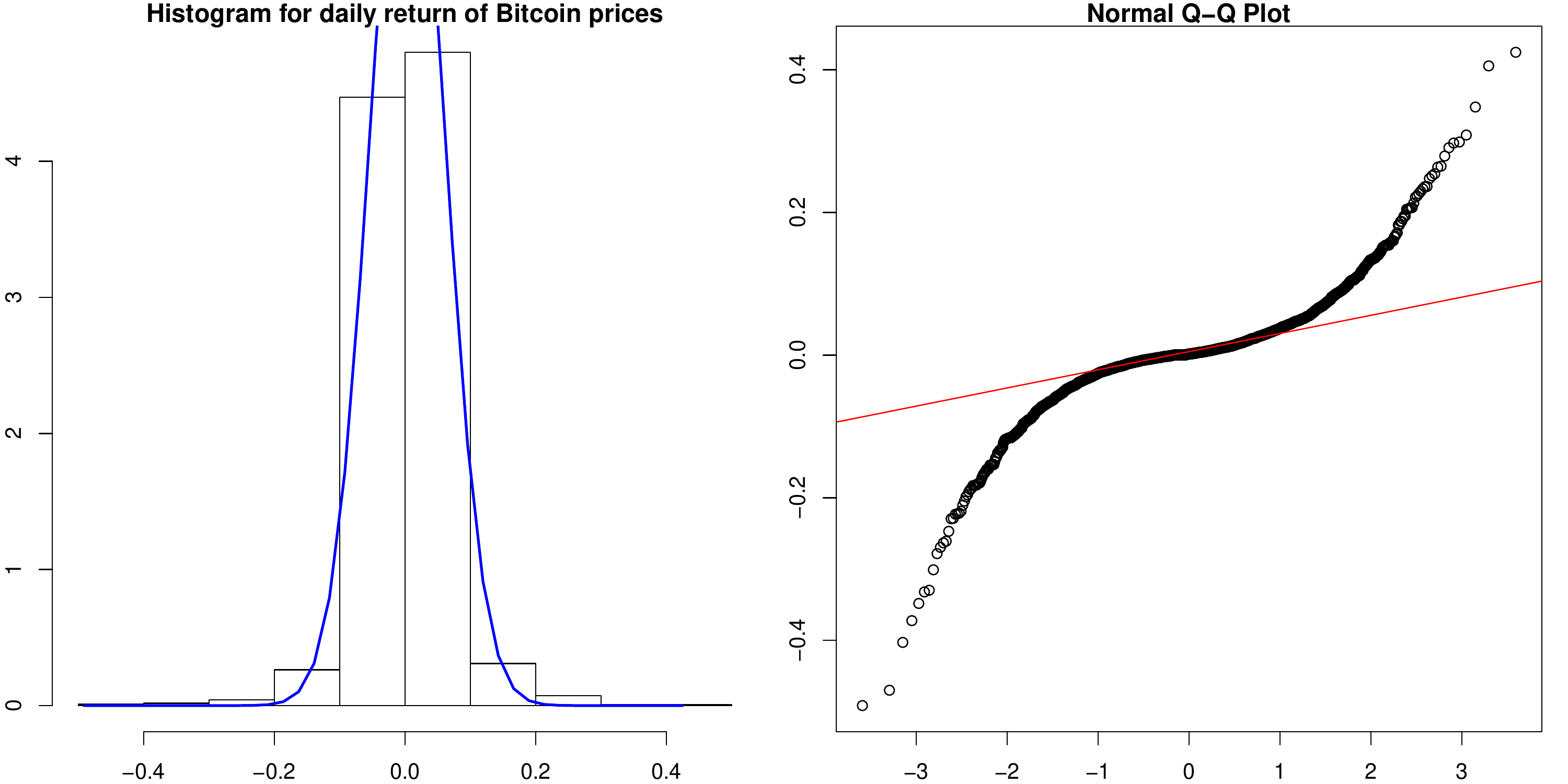}
\end{figure}

\begin{table}[]
\caption{Detected significant break periods for the Bitcoin log returns, the corresponding persistency parameter ($\hat{\alpha}_1+\hat{\beta}_1$) and the test statistics.  (***) means significant at both $0.95,0.90$.} \label{bit:break}
\begin{tabular}{l|l|l|l|l|c}
\hline\hline
  & $\hat{\tau}_1$   & $\hat{\tau}_2$    & in   & out& test statistics\\ \hline
1 & 2010/12/30 & 2011/10/15 & 1.14 & 1.00&11.68(***) \\ \hline
2 & 2013/04/14 & 2013/10/31 & 1.18 & 0.99&8.13(***) \\ \hline
3 & 2016/01/09 & 2017/12/09 & 1.05 & 0.99&19.04(***) \\ \hline\hline
\end{tabular}
\end{table}

%The research concludes that the GSADF test provides robust and accurate empirical results for long historical datasets that are subjected to multiple breaks and switching regimes.

%\clearpage

\section{Conclusion}
In this paper, we propose a supreme test for explosive GARCH process. Theoretical results on uniform parameter consistency and asymptotic distribution of the test statistics are provided.  Our test is easy to implement and can help to effectively identify explosive periods with good sizes and power via a simulation study. Moreover, we track the volatile period of the Apple stock and Bitcoin returns (log) in our application.  In further work one may extend the algorithm to online procedures allowing for real time detection of breaks.

\section{Technical proofs and lemmata}

\subsection{Existence of GARCH models}

\begin{proof}[Proof of Proposition \ref{prop:existencegarch}:]
    Following the proof of Lemma 2.3 in \cite{berkes2003}, there exists $m \in \IN$ such that $\IE \log |A_{m}(\theta)A_{m-1}(\theta)\dots A_{1}(\theta)|_2 < 0$. The function $[0,\delta) \to [0,\infty), s(t,\theta) = \IE |A_{m}(\theta)A_{m-1}(\theta)\dots A_{1}(\theta)|_2^t$ fulfills
    $s'(0,\theta) = \IE \log |A_{m}(\theta)A_{m-1}(\theta)\dots A_{1}(\theta)|_2 < 0$, thus $t \mapsto s(t,\theta)$ decreases in a neighborhood of $0$. Since $s(0) = 1$, this implies that there exists $0 < q < \delta$ such that
    \begin{equation}
        \IE |A_{m}(\theta)A_{m-1}(\theta)\dots A_{1}(\theta)|_2^{q} = s(q,\theta) < 1.\label{eq:matrixsmaller1}
    \end{equation}
    Define
\begin{eqnarray*}
    P_i(\theta) &:=& (X_i^2,\ldots, X_{i-r+1}^2, \sigma_{i}^2,\ldots, \sigma_{i-s+1}^2)',\\
    a_i(\theta) &:=& (\alpha_0 \zeta_i^2,0,\ldots,0,\alpha_0,0,\ldots,0)'.
\end{eqnarray*}
Following Section 3.1 in \cite{existencegarch}, the model \reff{eq:garchrecursion} admits the representation
\begin{equation}
    P_i(\theta) = A_i(\theta) P_{i-1}(\theta) + a_i(\theta).\label{eq:garchrecursion2}
\end{equation}
Therefore, $P_i(\theta) = G_{\zeta_i,\theta}(P_{i-1})$ with $G_{\zeta_i,\theta}(y) = A_i(\theta)\cdot y + a_i(\theta)$. Let $W_n(y,\theta) := G_{\zeta_n,\theta} \circ G_{\zeta_{n-1},\theta} \circ ... \circ G_{\zeta_1,\theta}(y)$. Then we have
\[
    W_n(y,\theta) - W_n(y',\theta) = A_n(\theta) A_{n-1}(\theta) \cdot ... \cdot A_1(\theta) \cdot (y-y').
\]
Using the submultiplicativity of $|\cdot|_2$, we therefore have with \reff{eq:matrixsmaller1} and some suitable constant $C > 0$:
\begin{eqnarray*}
    \| |W_n(y,\theta) - W_n(y',\theta)|_2 \|_q &\le& \||A_n(\theta) A_{n-1}(\theta) \cdot ... \cdot A_1(\theta)|_2\|_q |y-y'|_2\\
    &\le& C (s(q,\theta)^{q/m})^n.
\end{eqnarray*}
 By Theorem 2 in \cite{wushao2004}, we obtain existence and a.s. uniqueness of $X_i^2 = H(t,\sF_i)$, $\|X_0^2\|_q < \infty$ and $\delta_{q}^{X^2}(k) = O(c^k)$ with some $0 < c < 1$, i.e. (i) and (ii). (iii) is due to Proposition 1 in \cite{garch2004}.
\end{proof}

\subsection{Proofs for asymptotic theory}

Observe that
\[
    L_{n,\tau_1,\tau_2}^c(\theta) = L_{n,\tau_2}^c(\theta) - L_{n,\tau_1}^c(\theta), \quad\quad L_{n,r}^c(\theta) := \frac{1}{n}\sum_{i=1}^{\lfloor nr\rfloor}\ell(X_i^2,Y_i^c,\theta).
\]
Define $L(\theta) := \IE \ell(X_i^2,Y_i,\theta)$. This is well-defined due to $\IE\max\{-\ell(X_i^2,Y_i,\theta),0\} < \infty$ (cf. \cite{garch2004}, Proof of Theorem 2.1).

To prove Theorem \ref{theorem:consistency}, we introduce some notation. For some sequence of real-valued random variables $W_n$, we write $W_n \pto \infty$ if for each $M\in\IN$, $\IP(W_n < M) \to 0$ ($n\to\infty$). 

\begin{lemma}\label{lemma:stochastic_convergence_likelihood}
    Let Assumption \ref{ass1} hold. Let $\kappa > 0$, and $R := \{(\tau_1,\tau_2) \in [0,1]^2: \tau_1 < \tau_2, |\tau_1 - \tau_2|\ge \kappa\}$. For fixed $k\in \IN$, let $V_k(\theta) := \{\theta'\in \Theta:|\theta'-\theta|_1 < 1/k\}$. Define $W_i^{(k)}(\theta) := \inf_{\theta' \in V_k(\theta)}\ell(X_i^2,Y_i,\theta)$. Then:
    \begin{enumerate}
        \item[(i)] $\IE W_1^{(k)}(\theta) \in \IR \cup\{\infty\}$ and 
        \[
            \liminf_{n\to\infty}\inf_{(\tau_1,\tau_2)\in R}\frac{1}{n(\tau_2-\tau_1)}\sum_{i=\lfloor n\tau_1\rfloor+1}^{\lfloor n\tau_2\rfloor}W_{i}^{(k)}(\theta) \ge \IE W_1^{(k)}(\theta)\quad a.s.
        \]
        \item[(ii)] $L(\theta^{*})$ is finite and
        \[
            \sup_{(\tau_1,\tau_2)\in R}\frac{1}{n(\tau_2-\tau_1)}\sum_{i=\lfloor n\tau_1\rfloor+1}^{\lfloor n\tau_2\rfloor}\ell(X_i^2,Y_i,\theta^{*}) \to L(\theta^{*})\quad a.s.
        \]
    \end{enumerate}
\end{lemma}
\begin{proof}
    (i) Fix some $M \in\IN$. It holds that $\IE W_1^{(k)} \in \IR \cup\{\infty\}$ and $\IE[W_1^{(k)}\wedge M] \in \IR$ since $\IE\max\{-\ell(X_1^2,Y_1,\theta),0\} < \infty$. Define $S_m := \sum_{i=1}^{m}W_i^{(k)}(\theta) \wedge M$. By the ergodic theorem, we have
\[
    \lim_{m\to\infty}\frac{1}{m}S_m = \IE S_1 \quad a.s.
\]
It holds that $\lfloor n \tau_2 \rfloor \to \infty$ uniformly in $\tau_2 \in [\kappa,1]$, and thus
\begin{equation}
    \sup_{\tau_2\in[\kappa,1]}\big|\frac{S_{\lfloor n\tau_2\rfloor}}{\lfloor n\tau_2 \rfloor} - \IE S_1\big| \to 0\quad a.s.\label{eq:proof_stoch1}
\end{equation}
Furthermore we have that $\frac{S_{m}}{m}$ is a.s. bounded. We conclude that
\begin{equation}
    \sup_{\tau_1 \in [0,1]}\Big|\tau_1 \cdot \Big( \frac{S_{\lfloor n\tau_1\rfloor}}{\lfloor n\tau_1 \rfloor} - \IE S_1\Big)\Big| \to 0 \quad a.s.\label{eq:proof_stoch2}
\end{equation}
[Proof: Fix some $\omega \in \Omega$ and let $C > 0$ be such that $|\frac{S_{m}(\omega)}{m}| \le C$ for all $m\in\IN$. For $\varepsilon > 0$,
\[
    \sup_{\tau_1 \in [0,1]}\Big|\tau_1 \cdot \Big( \frac{S_{\lfloor n\tau_1\rfloor}}{\lfloor n\tau_1 \rfloor} - \IE S_1\Big)\Big| \le \sup_{\tau_1 \in [\frac{\varepsilon}{C},1]}\Big|\tau_1 \cdot \Big( \frac{S_{\lfloor n\tau_1\rfloor}}{\lfloor n\tau_1 \rfloor} - \IE S_1\Big)\Big| + \sup_{\tau_1 \in [0,\frac{\varepsilon}{C})}\Big|\tau_1 \cdot \Big( \frac{S_{\lfloor n\tau_1\rfloor}}{\lfloor n\tau_1 \rfloor} - \IE S_1\Big)\Big|
\]
The second term is bounded by $\frac{\varepsilon}{C}\cdot C = \varepsilon$, while for the first term we can choose $n$ large enough such that it is $\le \varepsilon$.]

With the decomposition
\[
    \frac{1}{n(\tau_2 - \tau_1)}\sum_{i=\lfloor n\tau_1\rfloor+1}^{\lfloor n\tau_2\rfloor}W_{i}^{(k)}(\theta) \wedge M = \frac{1}{\tau_2 - \tau_1}\Big[\frac{\lfloor n\tau_2\rfloor}{n}\cdot \frac{S_{\lfloor n\tau_2\rfloor}}{\lfloor n\tau_2\rfloor} - \frac{\lfloor n\tau_1\rfloor}{n\tau_1}\cdot \tau_1 \cdot \frac{S_{\lfloor n\tau_1\rfloor}}{\lfloor n \tau_1 \rfloor}\Big]
\]
and \reff{eq:proof_stoch1}, \reff{eq:proof_stoch2}, $\sup_{\tau_2 \in [\kappa,1]}|\frac{\lfloor n \tau_2\rfloor}{n}-\tau_2| \le n^{-1}$ we obtain
\[
    \inf_{(\tau_1,\tau_2) \in R}\frac{1}{n(\tau_2 - \tau_1)}\sum_{i=\lfloor n\tau_1\rfloor+1}^{\lfloor n\tau_2\rfloor}W_{i}^{(k)}(\theta) \wedge M \to \IE[W_1^{(k)}(\theta) \wedge M]\quad a.s.
\]
Since $W_i^{(k)}(\theta) \ge W_i^{(k)} \wedge M$ and applying $M\to\infty$ on the r.h.s., we obtain
\[
    \liminf_{n\to\infty}\inf_{(\tau_1,\tau_2) \in R}\frac{1}{n(\tau_2 - \tau_1)}\sum_{i=\lfloor n\tau_1\rfloor+1}^{\lfloor n\tau_2\rfloor}W_{i}^{(k)}(\theta) \ge \IE W_1^{(k)}(\theta)\quad a.s.
\]
(ii) The argument follows the same lines as (i). We obtain convergence since no truncation with $M$ is needed.
\end{proof}

\begin{proof}[Proof of Theorem \ref{theorem:consistency}] We make use of some results obtained in the proof of Theorem 2.1. in \cite{garch2004}. It was shown therein that $L(\theta) := \IE \ell(X_i^2,Y_i,\theta)$ fulfills
\begin{equation}
    \IE|\ell(X_i^2,Y_i,\theta^{*})| < \infty, \quad\quad \forall \theta \not= \theta^{*}: \quad L(\theta) > L(\theta^{*}).\label{eq:consistencyproofminus1}
\end{equation}
Let $k\in\IN$. Use the notation from Lemma \ref{lemma:stochastic_convergence_likelihood}. Let $\theta\not= \theta^{*}$. By Beppo-Levi's theorem, we have
\[
    \IE W_1^{(k)}(\theta) \uparrow L(\theta) > L(\theta^{*}).
\]
Thus, for each $\theta \not=\theta^{*}$ we can find $k(\theta)\in\IN$ such that $\IE W_1^{(k(\theta))}(\theta) > L(\theta^{*})$.\\

Let $\varepsilon > 0$ and $\Theta_{\varepsilon} := \{\theta \in \Theta: |\theta - \theta^{*}| \ge \varepsilon\}$. Then $\Theta_{\varepsilon}$ is compact, and there exist finitely many $\theta_{1},...,\theta_{l}$ with $\Theta_{\varepsilon} \subset \bigcup_{i=1}^{l}V_{k(\theta_i)}(\theta_i)$. Let
\[
    \delta := \min\{\inf_{i=1,...,l}\IE W_1^{(k(\theta_i))}(\theta_i) - L(\theta^{*}),1\} > 0.
\]

Suppose that $\sup_{(\tau_1,\tau_2)\in R}|\hat \theta_{n,\tau_1,\tau_2} - \theta^{*}| \ge \varepsilon$. By the minimal property of $\hat \theta_{n,\tau_1,\tau_2}$ and dividing by $\tau_2-\tau_1$, we conclude that
\begin{eqnarray}
    0 &\ge& \inf_{(\tau_1,\tau_2)\in R}\frac{1}{\tau_2-\tau_1}\big\{L_{n,\tau_1,\tau_2}^c(\hat \theta_{n,\tau_1,\tau_2}) - L_{n,\tau_1,\tau_2}^c(\theta^{*})\big\}\nonumber\\
    &=& \inf_{i=1,...,l}\inf_{(\tau_1,\tau_2)\in R} \inf_{\theta' \in V_{k(\theta_i)}(\theta_i)}\frac{1}{\tau_2-\tau_1}\{L_{n,\tau_1,\tau_2}^c(\theta') - L_{n,\tau_1,\tau_2}^c(\theta^{*})\}\nonumber\\
    &\ge& \inf_{i=1,...,l}\inf_{(\tau_1,\tau_2)\in R} \inf_{\theta' \in V_{k(\theta_i)}(\theta_i)}\frac{L_{n,\tau_1,\tau_2}^c(\theta')}{\tau_2-\tau_1} - \sup_{(\tau_1,\tau_2) \in R} \frac{L_{n,\tau_1,\tau_2}^c(\theta^{*})}{\tau_2-\tau_1}\label{eq:proofcons1}
\end{eqnarray}
We furthermore have:
\begin{eqnarray}
    &&\inf_{(\tau_1,\tau_2)\in R} \inf_{\theta' \in V_{k(\theta_i)}(\theta_i)}\frac{L_{n,\tau_1,\tau_2}^c(\theta')}{\tau_2-\tau_1}\nonumber\\
    &\ge& \inf_{(\tau_1,\tau_2)\in R}\inf_{\theta' \in V_{k(\theta_i)}(\theta_i)}\frac{L_{n,\tau_1,\tau_2}(\theta')}{\tau_2-\tau_1} - \kappa \cdot \sup_{(\tau_1,\tau_2)\in R}\sup_{\theta' \in \Theta}|L_{n,\tau_1,\tau_2}^c(\theta') - L_{n,\tau_1,\tau_2}(\theta')|\nonumber\\
    &\ge& \inf_{(\tau_1,\tau_2)\in R}\frac{1}{n(\tau_2-\tau_1)}\sum_{i=\lfloor n\tau_1\rfloor}^{\lfloor n\tau_2\rfloor}W_i^{(k(\theta_i))}(\theta_i)\nonumber\\
    &&\quad\quad - \kappa \cdot \sup_{(\tau_1,\tau_2)\in R}\sup_{\theta' \in \Theta}|L_{n,\tau_1,\tau_2}^c(\theta') - L_{n,\tau_1,\tau_2}(\theta')|.\label{eq:proofcons2}
\end{eqnarray}
By Lemma \ref{lemma:truncation}(i),
\begin{equation}
    \IP\Big(\sup_{(\tau_1,\tau_2)\in R}\sup_{\theta' \in \Theta}|L_{n,\tau_1,\tau_2}^c(\theta') - L_{n,\tau_1,\tau_2}(\theta')| > \frac{\delta}{8}\Big) = o(1).\label{eq:proofcons3}
\end{equation}
By Lemma \ref{lemma:stochastic_convergence_likelihood}(i),
\begin{equation}
    \IP\Big(\inf_{i=1,...,l}\inf_{(\tau_1,\tau_2)\in R}\frac{1}{n(\tau_2-\tau_1)}\sum_{i=\lfloor n\tau_1\rfloor}^{\lfloor n\tau_2\rfloor}W_i^{(k(\theta_i))}(\theta_i) \le L(\theta^{*})+\frac{\delta}{2}\Big) = o(1).\label{eq:proofcons4}
\end{equation}
By Lemma \ref{lemma:stochastic_convergence_likelihood}(ii),
\begin{equation}
    \IP\Big(\sup_{(\tau_1,\tau_2)\in R}\frac{L_{n,\tau_1,\tau_2}^c(\theta^{*})}{\tau_2-\tau_1} \ge L(\theta^{*})+\frac{\delta}{8}\Big) = o(1).\label{eq:proofcons5}
\end{equation}
Inserting \reff{eq:proofcons3}, \reff{eq:proofcons4} into \reff{eq:proofcons2} and afterwards using \reff{eq:proofcons5}, we have
\begin{eqnarray*}
    && \IP\Big(\sup_{(\tau_1,\tau_2)\in R}|\hat \theta_{n,\tau_1,\tau_2} - \theta^{*}| \ge \varepsilon\Big)\\
    &\le& \IP\Big( 0 \ge \inf_{i=1,...,l}\inf_{(\tau_1,\tau_2)\in R} \inf_{\theta' \in V_{k(\theta_i)}(\theta_i)}\frac{L_{n,\tau_1,\tau_2}^c(\theta')}{\tau_2-\tau_1} - \sup_{(\tau_1,\tau_2) \in R} \frac{L_{n,\tau_1,\tau_2}^c(\theta^{*})}{\tau_2-\tau_1}\Big)\\
    &\le& \IP\Big(0 \ge (L(\theta^{*}) + \frac{\delta}{2}) - \frac{\delta}{8} - (L(\theta^{*}) + \frac{\delta}{8}) = \frac{\delta}{4}\Big) + o(1) = o(1),
\end{eqnarray*}
showing the assertion.
\end{proof}

\begin{proof}[Proof of Theorem \ref{theorem:bahadur}]
    Let $\kappa > 0$ und define $R := \{(\tau_1,\tau_2) \in [0,1]^2: \tau_1 < \tau_2, |\tau_1 - \tau_2| \ge \kappa\}$. By Theorem \ref{theorem:consistency}, we have
    \begin{equation}
        \sup_{(\tau_1,\tau_2) \in R}|\hat \theta_{n,\tau_1,\tau_2}- \theta^{*}|_1 \to 0\quad a.s.\label{eq:bahadurproof_1}
    \end{equation}
    Therefore, $\hat \theta_{n,\tau_1,\tau_2} \in int(\Theta)$ uniformly in $(\tau_1,\tau_2) \in R$ for $n$ large enough. Thus there exists $\bar \Theta \subset int(\Theta)$ with $\hat \theta_{n,\tau_1,\tau_2} \in \bar \Theta$ (for $n$ large enough, $(\tau_1,\tau_2) \in R$.
    
    By a Taylor expansion, we have
    \begin{eqnarray}
        \hat \theta_{n,\tau_1,\tau_2} - \theta^{*} &=& -[\nabla_{\theta}^2 L_{n,\tau_1,\tau_2}^c(\bar\theta)]^{-1}\nabla_{\theta}L_{n,\tau_1,\tau_2}^c(\theta^{*})\nonumber\\
        &=& -[(\tau_2-\tau_1)V(\theta^{*}) + T_{n,\tau_1,\tau_2}(\bar \theta_{n,\tau_1,\tau_2})]^{-1}\nabla_{\theta}L_{n,\tau_1,\tau_2}^c(\theta^{*}),\label{eq:bahadurproof_2}
    \end{eqnarray}
    where $T_{n,\tau_1,\tau_2}(\bar \theta_n) = \nabla_{\theta}^2 L_{n,\tau_1,\tau_2}^c(\bar\theta_n) - (\tau_2 - \tau_1)V(\theta^{*})$ and $\bar \theta_{n,\tau_1,\tau_2} \in \Theta$ with $|\bar \theta_{n,\tau_1,\tau_2} - \theta^{*}|_1 \le |\hat \theta_{n,\tau_1,\tau_2} - \theta^{*}|_1$. By Lemma \ref{lemma:truncation} and Lemma \ref{lemma:convergence} and $|\frac{\lfloor n\tau_2\rfloor - \lfloor n \tau_1\rfloor}{n}-(\tau_1-\tau_2)| \le 2n^{-1}$, we have
    \begin{equation}
        \sup_{(\tau_1,\tau_2) \in R}|T_{n,\tau_1,\tau_2}(\bar \theta_{n,\tau_1,\tau_2})|_{1} \le \sup_{(\tau_1,\tau_2)\in R}|V(\bar \theta_{n,\tau_1,\tau_2}) - V(\theta^{*})|_1 + O_p( \log(n)^{3/2}n^{-1/2}).\label{eq:bahadurproof_3}
    \end{equation}
    By Lemma \ref{lem:analytical_properties}(ii) applied to $p = q$ and $\bar \Theta$, we obtain $\iota > 0$, $C > 0$ and $\rho \in (0,1)$ such that for $|\theta - \theta^{*}|_1 < \iota$,
    \begin{eqnarray}
        |V(\theta) - V(\theta^{*})|_1 &\le& C(1 + \||Y_i|_{(\rho^j)_j,q}^{q}\|_1)(1 + \IE \zeta_0^2)\cdot |\theta - \theta^{*}|_1\nonumber\\
        &\le& C(1 + \frac{D^q}{1-\rho})(1 + \IE \zeta_0^2) \cdot |\theta - \theta^{*}|_1 =: \tilde C \cdot |\theta - \theta^{*}|_1.\label{eq:bahadurproof_4}
    \end{eqnarray}
    Since $\IE \nabla_{\theta}\ell(Y_i,X_i^2,\theta^{*}) = 0$ (cf. Proposition \ref{prop:matrices}(iii)), Lemma \ref{lemma:truncation} and Lemma \ref{lemma:convergence}, we have
    \begin{equation}
        \sup_{(\tau_1,\tau_2)\in R}|\nabla_{\theta}L_{n,\tau_1,\tau_2}^c(\theta^{*})|_1 = O_p( \log(n)^{3/2}n^{-1/2}).\label{eq:bahadurproof_5}
    \end{equation}
    Inserting \reff{eq:bahadurproof_1} into \reff{eq:bahadurproof_4}, we obtain $\sup_{(\tau_1,\tau_2) \in R}|T_{n,\tau_1,\tau_2}(\bar \theta_{n,\tau_1,\tau_2})|_2 = o_p(1)$. From \reff{eq:bahadurproof_2} and \reff{eq:bahadurproof_5} we obtain 
    \begin{equation}
        \sup_{(\tau_1,\tau_2) \in R}|\hat \theta_{n,\tau_1,\tau_2} - \theta^{*}|_1 = O_p(\log(n)^{3/2}n^{-1/2}).\label{eq:bahadurproof_6}
    \end{equation}
    By \reff{eq:bahadurproof_2}, we have
    \begin{eqnarray}
        &&\big|\hat \theta_{n,\tau_1,\tau_2} - \theta^{*} + ((\tau_2-\tau_1)V(\theta^{*}))^{-1}\cdot \nabla_{\theta}L_{n,\tau_1,\tau_2}^{c}(\theta^{*})\big|_2\nonumber\\
        &\le& |((\tau_2 - \tau_1)V(\theta^{*}) + T_{n,\tau_1,\tau_2}(\bar \theta_{n,\tau_1,\tau_2}))^{-1}|_1 \cdot |T_{n,\tau_1,\tau_2}(\bar \theta_{n,\tau_1,\tau_2})|_1\nonumber\\
        &&\quad\quad\times |((\tau_2 - \tau_1)V(\theta^{*}))^{-1}\nabla_{\theta}L_{n,\tau_1,\tau_2}^{c}(\theta^{*})|_1.\label{eq:bahadurproof_7}
    \end{eqnarray}
    Using \reff{eq:bahadurproof_3}, \reff{eq:bahadurproof_4} and \reff{eq:bahadurproof_6}, we have $\sup_{(\tau_1,\tau_2)\in R}|T_{n,\tau_1,\tau_2}(\bar \theta_{n,\tau_1,\tau_2})|_1 = O_p(\log(n)^{3/2}n^{-1/2})$. Inserting this and \reff{eq:bahadurproof_5} into \reff{eq:bahadurproof_7}, we obtain
    \[
        \sup_{(\tau_1,\tau_2)\in R}\big|\hat \theta_{n,\tau_1,\tau_2} - \theta^{*} + ((\tau_2 - \tau_1)V(\theta^{*}))^{-1}\cdot \nabla_{\theta}L_{n,\tau_1,\tau_2}^{c}(\theta^{*})\big|_1 = O_p(\log(n)^3 n^{-1}).
    \]
    Since $\sup_{(\tau_1,\tau_2)\in R}|\nabla_{\theta}L_{n,\tau_1,\tau_2}^c(\theta^{*}) -\nabla_{\theta}L_{n,\tau_1,\tau_2}(\theta^{*})|_1 = O_p(n^{-1})$ by Lemma \ref{lemma:truncation}, the proof is finished.
\end{proof}

\begin{proof}[Proof of Theorem \ref{theorem:gaussian}] Let $R := \{(\tau_1,\tau_2) \in [0,1]^2: \tau_1 < \tau_2, |\tau_1 - \tau_2| \ge \kappa\}$. By Lemma \ref{lemma:dependence} (applied with $M = 2+\frac{a'}{4}$, $a = \frac{a'}{4}$), we obtain $C > 0$, $\rho \in (0,1)$ and $\iota > 0$ such that
\[
    \delta_{M}^{\nabla_{\theta}\ell(Z,\theta^{*})}(k) \le C \rho^k,
\]
and thus (component-wise) $\Delta_{M}^{\nabla_{\theta}\ell(Z,\theta^{*})}(k) \le \frac{C}{1-\rho}\rho^k$.
%By Theorem 2.1 and Corollary 2.1(iii) in \cite{berkesliuwu2014} applied with $p = M$ therein, we obtain that there exists a probability space $(\Omega_c,\sA_c,\IP_c)$ and random variables $V_i^c$, $i = 1,...,n$ and a standard Brownian motion $\IB_c(\cdot)$ defined on this space such that
Let $W_i := -V(\theta^{*})^{-1}\nabla_{\theta}\ell(X_i^2,Y_i,\theta^{*})$ and $S(j) := \sum_{i=1}^{j}W_i$. By Proposition \ref{prop:matrices}(iii), $\IE W_i = 0$ and
\begin{eqnarray*}
    \Sigma &:=& \Cov(W_i)= V(\theta^{*})^{-1}I(\theta^{*})V(\theta^{*})^{-1}=\frac{\mu_4-1}{2}V(\theta^{*})^{-1}.
\end{eqnarray*}
By Corollary 1 in \cite{zhouwu2011}, there exists a richer probability space and i.i.d. $V_1,V_2,\ldots \sim N(0, I_{(r+s+1)\times (r+s+1)})$, a process $(\hat S(i))_{i=1,...,n}$ and $S^{0}(i) = \sum_{j=1}^{i} V_j$ such that $(S(i))_{i=1,\ldots,n} \overset{d}{=} (\hat S(i))_{i=1,\ldots,n}$ and
\[
    \max_{i=1,\ldots,n}|\hat S(i) - \Sigma^{1/2} S^{0}(i)| = O_{p}(n^{1/\min\{M,4\}}\log(n)^{3/2}).
\]
With Theorem \ref{theorem:bahadur} we obtain:
\begin{equation}
    \sup_{(\tau_1,\tau_2)\in R}\big|\sqrt{n}(\tau_2 - \tau_1)(\hat \theta_{n,\tau_1,\tau_2} - \theta^{*}) - n^{-1/2}(S(\lfloor n\tau_2\rfloor) - S(\lfloor n\tau_1\rfloor))\big| = O_p(\log(n)^3 n^{-1/2}).\label{eq:proof_gaussian_eq1}
\end{equation}
By the Gaussian approximation result above, on $D(R)^{r+s+1}$
\begin{equation}
    n^{-1/2}(S(\lfloor n\tau_2\rfloor) - S(\lfloor n\tau_1\rfloor)) \overset{d}{=} n^{-1/2}(\hat S(\lfloor n\tau_2\rfloor) - \hat S(\lfloor n\tau_1\rfloor))\label{eq:proof_gaussian_eq2}
\end{equation}
and
\begin{eqnarray}
    &&\sup_{(r_1,r_2)\in R}\big|n^{-1/2}(\hat S(\lfloor n\tau_2\rfloor) - \hat S(\lfloor n\tau_1\rfloor)) - \Sigma^{1/2}\cdot n^{-1/2}(S^{0}(\lfloor n\tau_2\rfloor) - S^{0}(\lfloor n\tau_1\rfloor))\big|\nonumber\\
    &=& O_p(n^{\frac{1}{\min\{M,4\}}-\frac{1}{2}}\log(n)^{3/2}).\label{eq:proof_gaussian_eq3}
\end{eqnarray}

%%%%%%%%old
%\ignore{
%\begin{eqnarray*}
 %   &&\sqrt{n}\sup_{(\tau_1,\tau_2)\in R}\big\{ \hat \theta_{n,\tau_1,\tau_2} - \theta^{*}\big\}\\
  %  &=& \sqrt{n}\sup_{(\tau_1,\tau_2)\in R}\big\{-V(\theta^{*})^{-1}\nabla_{\theta}L_{n,\tau_1,\tau_2}(\theta^{*})\big\} + O_p(\log(n)^3 n^{-1/2})\\
   % &=& n^{-1/2}\sup_{(\tau_1,\tau_2)\in R}\big\{S(\lfloor n\tau_2\rfloor) - S(\lfloor n\tau_1\rfloor)\big\} + O_p(\log(n)^3 n^{-1/2})\\
    %&\overset{d}{=}& n^{-1/2}\sup_{(\tau_1,\tau_2)\in R}\big\{\hat S(\lfloor n\tau_2\rfloor) - \hat S(\lfloor n\tau_1\rfloor)\big\} + O_p(\log(n)^3 n^{-1/2})\\
    %&=& n^{-1/2}\Sigma^{1/2} \sup_{(\tau_1,\tau_2)\in R}\big\{S^{0}(\lfloor n\tau_2\rfloor) - S^0(\lfloor n\tau_1\rfloor)\big\}\\
    %&&\quad\quad + O_p(\log(n)^3 n^{-1/2}) + O_p(n^{\frac{1}{\min\{M,4\}}-\frac{1}{2}}\log(n)^{3/2}).
%\end{eqnarray*}
%}
%%%%%%%%old

By Donsker's theorem, it holds in $D[0,1]^{r+s+1}$ that
$n^{-1/2}S^{0}(\lfloor nr\rfloor) \dto B(r)$ with some standard $(r+s+1)$-dimensional Brownian motion $B$. Applying the continuous mapping theorem, we obtain in $D(R)^{r+s+1}$:
\begin{equation}
    n^{-1/2}\Sigma^{1/2}\big\{S^{0}(\lfloor n\tau_2\rfloor) - S^0(\lfloor n\tau_1\rfloor)\big\} \dto \Sigma^{1/2}\big\{B(\tau_2) - B(\tau_1)\big\}.\label{eq:proof_gaussian_eq4}
\end{equation}
Combining \reff{eq:proof_gaussian_eq1}, \reff{eq:proof_gaussian_eq2}, \reff{eq:proof_gaussian_eq3} and \reff{eq:proof_gaussian_eq4}, we obtain the result.
\end{proof}

\begin{proof}[Proof of Proposition \ref{prop:convergence_matrices}]
    Using similar arguments as in the proof of Theorem \ref{theorem:consistency} (now with $1 - (\tau_2 - \tau_1) \ge \kappa$ instead of $\tau_2 - \tau_1 \ge \kappa$), we obtain
    \begin{equation}
        \sup_{(\tau_1,\tau_2)\in \bar R_{\kappa}}|\bar \theta_{n,\tau_1,\tau_2} - \theta^{*}| \pto 0.\label{eq:consistency_borderestimate}
    \end{equation}
    (i) By Lemma \ref{lemma:truncation} and Lemma \ref{lemma:convergence} and $|\frac{1 - (\lfloor n\tau_2\rfloor - \lfloor n\tau_1\rfloor)}{1-(\tau_2 - \tau_1)} - 1|\le n^{-1}$, we have
    \[
        \sup_{(\tau_1,\tau_2)\in \bar R_{\kappa}}|\bar V_{n,\tau_1,\tau_2}(\theta) - (1-(\tau_2-\tau_1))V(\theta)|_1 \pto 0.
    \]
    By Lipschitz continuity of $V(\cdot)$ (cf. \reff{eq:bahadurproof_4}) and \reff{eq:consistency_borderestimate}, we obtain the result.\\
    (ii) Define $g(x,y,\theta) := \nabla_{\theta}\ell(x,y,\theta)\cdot \nabla_{\theta}\ell(x,y,\theta)'$ and $\tilde g_{\tilde \theta}(\zeta,y,\theta) := g(R_{\zeta}(y,\tilde \theta),y,\theta)$, where $R_{\zeta}(y,\theta) := \zeta^2 \sigma(y,\theta)^2$. Let $\bar\Theta \subset int(\Theta)$ be some compact set. Using Lemma \ref{lem:analytical_properties}(ii) and \reff{eq:likelihood_upperbound}, it is easy to see that for any $p > 0$, one can find $\iota > 0$, $C > 0$ and $\rho\in(0,1)$ such that (component-wise),
    \begin{eqnarray}
        &&\sup_{\theta, \tilde \theta \in \bar\Theta, |\theta - \tilde \theta|_1 < \iota}|\tilde g_{\tilde \theta}(\zeta,y,\theta) -  \tilde g_{\tilde \theta}(\zeta,y',\theta)|\nonumber\\
        &\le& C (1 +  |y|_{(\rho^j)_j,2p}^{2p} + |y'|_{(\rho^j)_j,2p}^{2p})|y-y'|_{(\rho^{j})_j,p}^p (1+\zeta^2)^2.\label{eq:i_xdiff}
    \end{eqnarray}
    and
    \begin{equation}
        \sup_{\theta, \theta', \tilde\theta \in \bar\Theta, |\theta - \tilde \theta|_1 < \iota, |\theta' - \tilde \theta|_1 < \iota}\frac{|\tilde g_{\tilde \theta}(\zeta,y,\theta) - \tilde g_{\tilde \theta}(\zeta,y,\theta')|}{|\theta - \theta'|_1} \le C(1+|y|_{(\rho^j)_j,p}^p)(1+\zeta^2)^2.\label{eq:i_thetadiff}
    \end{equation}
    In the following we will enlarge $C,\rho$ and reduce $\iota$ if necessary without further notice. Note that $\sup_{\theta \in \Theta}|\nabla_{\theta}(\sigma(0,\theta)^2)| <\infty$ and thus (component-wise) $\sup_{\theta \in \Theta}|\nabla_{\theta}\ell(x,0,\theta)| \le C(1+|x|)$. With Lemma \ref{lem:analytical_properties}(iii) we conclude that (component-wise) $\sup_{\theta \in \Theta}|\nabla_{\theta}\ell(x,y,\theta)| \le C(1 +  |y|_{(\rho^j)_j,1}^{2})(1+|x|)\cdot |y|_{(\rho^j)_j,1}$. Using again Lemma \ref{lem:analytical_properties}(iii), we obtain
    \[
        \sup_{\theta \in \Theta}|g(x,y,\theta) - g(x,y',\theta)| \le C (1 +  |y|_{(\rho^j)_j,1}^{5} + |y'|_{(\rho^j)_j,1}^{5})(1+|x|)^2\cdot |y-y'|_{(\rho^{j})_j,1}.
    \]
    This shows that $g$ has similar properties as $\nabla_{\theta}^2\ell$, but with factors $(1+\zeta^2)^2$ in \reff{eq:i_xdiff}, \reff{eq:i_thetadiff} instead of $(1+\zeta^2)$. Therefore, we obtain the same result as in (i) under the stated moment condition.
\end{proof}

\subsection{Technical lemmata}

\begin{proof}[Proof of Proposition \ref{prop:matrices}] (i) By Proposition \ref{prop:existencegarch},  there exists $q > 0$ with $\|X_0^2\|_q < \infty$. From the bounds \reff{eq:likelihood_upperbound} (applied with $p = q$) we conclude that $V(\theta)$, $I(\theta)$ are finite as long as $|\theta - \theta^{*}|$ is small enough.\\
(ii),(iii) This was already shown in \cite{garch2004}, see the proof step (ii) of Theorem 2.2 (the missing $\frac{1}{2}$ is due to the different formulation of the likelihood).
\end{proof}

\begin{lemma}[Negligibility of truncation]\label{lemma:truncation}
    Let Assumption \ref{ass1} hold. Then for $g = \nabla_{\theta}^l\ell$, $l = 0,1,2$ it holds that
    \begin{enumerate}
        \item[(i)]
    \[
        \sup_{r\in[0,1]}\sup_{\theta \in \Theta}\Big|\sum_{i=1}^{\lfloor n r\rfloor}g(X_i^2,Y_i^c,\theta) - \sum_{i=1}^{\lfloor n r\rfloor}g(X_i^2,Y_i,\theta)\Big| = O_{p}(1).
    \]
        \item[(ii)]
        \[
        \sup_{r\in[0,1]}\sup_{\theta \in \Theta}\frac{1}{n}\Big|\sum_{i=1}^{\lfloor n r\rfloor}g(X_i^2,Y_i^c,\theta) - \sum_{i=1}^{\lfloor n r\rfloor}g(X_i^2,Y_i,\theta)\Big| = 0\quad a.s.
    \]
    \end{enumerate}
\end{lemma}
\begin{proof}[Proof of Lemma \ref{lemma:truncation}]
    Note that for arbitrary $0 < \tilde q \le \min\{q,1\}$ and random variables $Z_i = (Z_{i1},Z_{i2},...)$ with $\|Z_i\|_{q} \le D$ it holds that
    \[
        \||Z_i|_{(\rho^j)_j,1}\|_{\tilde q} \le \Big(\sum_{j=1}^{\infty}\rho^{\tilde q j} \|Z_{ij}\|_{\tilde q}^{\tilde q}\Big)^{1/\tilde q} \le D (\frac{1}{1-\rho^{\tilde q}})^{1/\tilde q} =: \tilde D(\tilde q).
    \]
    Let
    \[
        W_i := \sup_{\theta \in \Theta}|g(X_i^2,Y_i^c,\theta) - g(X_i^2,Y_i,\theta)|.
    \]
    By Lemma \ref{lem:analytical_properties}(iii), we have with H\"{o}lder's inequality for $0 < q' \le q$ chosen such that $0 < q'(l+3) \le 1$:
    \begin{eqnarray}
        && \|W_i\|_{q'}\nonumber\\
        &=& \big\|\sup_{\theta \in \Theta}|g(X_i^2,Y_i^c,\theta) - g(X_i^2,Y_i,\theta)|\big\|_{q'} \nonumber\\
        &\le& C (1 +  \||Y_i|_{(\rho^j)_j,1} \|_{q'(l+3)}^{l+1} + \||Y_i^c|_{(\rho^j)_j,1} \|_{q'(l+3)}^{l+1})\cdot (1 + \|X_i^2\|_{q'(l+3)}) \| |y-y'|_{(\rho^{j})_j,1}\|_{q'(l+3)}\nonumber\\
        &\le& C(1+2\tilde D(q')^{l+1})(1+D)\cdot \Big(\sum_{j=i}^{\infty}\rho^{q'(l+3)j}\|X_j^2\|_{q'(l+3)}^{q'(l+3)}\Big)^{1/(q'(l+3))}\nonumber\\
        &\le& C(1+2\tilde D(q')^{l+1})(1+D)\tilde D(q'(l+3)) \rho^{i} =: \tilde C \cdot \rho^{i}.\label{eq:proof_lemma_truncation1}
    \end{eqnarray}
    Therefore, we have
    \begin{eqnarray*}
        \Big\|\sup_{r\in[0,1]}\sup_{\theta \in \Theta}\Big|\sum_{i=1}^{\lfloor n r\rfloor}g(X_i^2,Y_i^c,\theta) - \sum_{i=1}^{\lfloor n r\rfloor}g(X_i^2,Y_i,\theta)\Big|\Big\|_{q'} &\le& \Big\|\sum_{i=1}^{n}W_i\Big\|_{q'}\\
        &\le& \Big(\sum_{i=1}^{n}\big\|W_i\big\|_{q'}^{q'}\Big)^{1/q'}\\
        &\le& \tilde C \Big(\sum_{i=1}^{n}(\rho^{q'})^{i}\Big)^{1/q'} < \infty,
    \end{eqnarray*}
    giving the result.\\
    (ii) It holds that
    \[
        \sup_{r\in[0,1]}\sup_{\theta \in \Theta}\frac{1}{n}\Big|\sum_{i=1}^{\lfloor n r\rfloor}g(X_i^2,Y_i^c,\theta) - \sum_{i=1}^{\lfloor nr\rfloor}g(X_i^2,Y_i,\theta)\Big| \le \frac{1}{n}\sum_{i=1}^{n}W_i.
    \]
    In the following we show that $W_i \to 0$ ($i\to\infty$) a.s.. Then the assertion follows with a Cesaro sum argument. Let $\varepsilon > 0$ be arbitrary. Then with Markov's inequality and \reff{eq:proof_lemma_truncation1},
    \[
        \sum_{i=1}^{\infty}\IP(|W_i| > \varepsilon) \le \sum_{i=1}^{\infty}\frac{\tilde C^{q'}}{\varepsilon^{q'}}\rho^{q' i} < \infty,
    \]
    showing $W_i \to 0$ with Borel-Cantelli's lemma.
\end{proof}

Let us use the abbreviation $Z_i := (X_i^2, Y_i)$. 
We now state results about the dependence measure of the stationary processes $g(Z_i,\theta)$, where $g \in \{\nabla_{\theta}\ell, \nabla_{\theta}^2\ell\}$.

\begin{lemma}[Dependence measures of $\nabla_{\theta}\ell$, $\nabla_{\theta}^2\ell$]\label{lemma:dependence} Let Assumption \ref{ass1} hold. Let $M \ge 1$. Assume that $\IE |\zeta_0|^{2(M+a)} < \infty$ for some $a > 0$. Let $g \in \{\nabla_{\theta}\ell, \nabla_{\theta}^2\ell\}$. Then there exists some $C > 0$, $\rho \in (0,1)$, $\iota > 0$ such that
\[
    \sup_{|\theta - \theta^{*}|_{1} < \iota}\delta_{M}^{g(Z,\theta)}(k) \le C \rho^k, \quad\quad \delta_{M}^{\sup_{|\theta - \theta^{*}|_{1} < \iota}g(Z,\theta)}(k) \le C \rho^k.
\]
\end{lemma}
\begin{proof}[Proof of Lemma \ref{lemma:dependence}]
    We only prove the second assertion, the first is nearly the same. Let $(X_i^2)^{*} = H(\sF_i^{*})$ and $Z_i^{*} := ( (X_i^{2})^{*},Y_i^{*})$. Let $\kappa = \frac{a}{3}$. Choose $p > 0$ small enough such that
     $\frac{(M+3\kappa)}{\kappa}p \le q$. By Hoelder's inequality ($\frac{M}{M+3\kappa} + \frac{\kappa}{M+3\kappa} + \frac{2\kappa}{M+3\kappa} = 1$) and Lemma \ref{lem:analytical_properties}(ii) there exists $\iota > 0$, $C > 0$, $\rho \in (0,1)$ such that
    \begin{eqnarray*}
        && \delta_{M}^{\sup_{|\theta - \theta^{*}|_1 < \iota}|g(Z,\theta)}(i)|\\
        &=& \Big\|\sup_{|\theta - \theta^{*}|_1 < \iota}|g(Z_i,\theta)| - \sup_{|\theta - \theta^{*}|_1 < \iota}|g(Z_i^{*},\theta)|\Big\|_{M}\\
        &\le& \Big\|\sup_{|\theta - \theta^{*}|_1 < \iota}|g(Z_i,\theta)-g(Z_i^{*},\theta)|\Big\|_{M}\\
        &=& \Big\| \sup_{|\theta - \theta^{*}|_1 < \iota}|\tilde g_{\theta^{*}}(\zeta_i,Y_i,\theta)-\tilde g_{\theta^{*}}(\zeta_i,Y_i^{*},\theta)|\Big\|_{M}\\
        &\le& C (1 +  \| |Y_i|_{(\rho^j)_j,2p}^{2p}\|_{\frac{M+3\kappa}{2\kappa}} + \| |Y_i^{*}|_{(\rho^j)_j,2p}^{2p}\|_{\frac{M+3\kappa}{2\kappa}})\| |Y_i-Y_i^{*}|_{(\rho^{j})_j,p}^p\|_{\frac{M+3\kappa}{\kappa}} (1+\|\zeta^2\|_{M+3\kappa})\\
        &\le& C (1 + 2\frac{D^{2p}}{1-\rho}) (1 + \|\zeta^2\|_{M+3\kappa})\cdot \sum_{j=1}^{\infty}\rho^j \|X_{i-j}^2 - (X_{i-j}^{*})^2\|^p_{\frac{(M+3\kappa)p}{\kappa}}\\
        &\le& \tilde C \cdot \sum_{j=1}^{i}\rho^j [\delta^{X^2}_q(i-j)]^p,
    \end{eqnarray*}
    where $\tilde C :=C (1 + 2\frac{D^{2p}}{1-\rho}) (1 + \|\zeta^2\|_{M+3\kappa})$. By Proposition \ref{prop:existencegarch}(ii), it holds that $\delta^{X^2}_q(k) = O(c^k)$, which finishes the proof.
\end{proof}

In the following we make use of results from \cite{zhangwu2017}. Therefore we have to define $\Delta_q^{Z}(m) := \sum_{k=m}^{\infty}\delta_q^{Z}(k)$ and $\|Z\|_{q,\alpha} := \sup_{m\ge 0}(m+1)^{\alpha}\Delta_q^{Z}(m)$.

\begin{lemma}\label{lemma:convergence}
    Let Assumption \ref{ass1} hold. Additionally, assume that for some $a' > 0$, $\IE |\zeta_0|^{4+a'} < \infty$. Then there exists $\iota > 0$ such that for $g = \nabla_{\theta}^l\ell$, $l = 1,2$, it holds that
    \[
        \sup_{|\theta - \theta^{*}|_1 < \iota}\sup_{r \in [0,1]}\Big|\frac{1}{n}\sum_{i=1}^{\lfloor nr\rfloor}\big\{g(X_i^2,Y_i,\theta) - \IE g(X_i^2,Y_i,\theta)\big\}\Big| = O_{p}\Big(\Big(\frac{\log(n)^3}{n}\Big)^{1/2}\Big).
    \]
\end{lemma}
\begin{proof}[Proof of Lemma \ref{lemma:convergence}]
    Let $\iota > 0$ (is chosen below). Let $S_j(\theta) := \sum_{i=1}^{j}\big\{g(X_i^2,Y_i,\theta) - \IE g(X_i^2,Y_i,\theta)\big\}$, $j = 1,...,n$. For fixed $n\in\IN$, choose $d \in \IN$ such that $2^{d-1} \le n \le 2^{d}$. For $i = 0,1,...,d-1$, define
    \[
        \Phi_i(\theta) := \max_{1 \le k \le 2^{d-i}}|S_{2^i\cdot k}(\theta) - S_{2^{i}(k-1)}|.
    \]
    By a dyadic expansion of $j\in\{1,...,n\}$ we obtain
    \[
        \max_{j=0,...,n}|S_j(\theta)| \le \sum_{i=0}^{d-1}\Phi_i(\theta).
    \]
    Note that
    \begin{eqnarray*}
        &&\sup_{|\theta - \theta^{*}|_1 < \iota}\sup_{r \in [0,1]}\Big|\sum_{i=1}^{\lfloor nr\rfloor}\big\{g(X_i^2,Y_i,\theta) - \IE g(X_i^2,Y_i,\theta)\big\}\Big|\\
        &\le& \sum_{i=0}^{d-1}\sup_{|\theta - \theta^{*}|_1 < \iota}\Phi_i(\theta).
    \end{eqnarray*}
    Thus, for $Q > 0$, by stationarity,
    \begin{eqnarray}
        && \IP\Big(\sup_{|\theta - \theta^{*}|_1 < \iota}\sup_{r \in [0,1]}\Big|\sum_{i=1}^{\lfloor nr\rfloor}\big\{g(X_i^2,Y_i,\theta) - \IE g(X_i^2,Y_i,\theta)\big\}\Big| > Q (n\log(n)^3)^{1/2}\Big)\nonumber\\
        &\le& \sum_{i=0}^{d-1}\IP\Big(\sup_{|\theta - \theta^{*}|_1 < \iota}\Phi_i(\theta) > \frac{Q (n\log(n)^3)^{1/2}}{d}\Big)\nonumber\\
        &\le & \sum_{i=0}^{d-1}2^{d-i}\cdot \IP\Big(\sup_{|\theta - \theta^{*}|_1 < \iota} |S_{2^i}(\theta)| > \frac{Q (n\log(n)^3)^{1/2}}{d}\Big).\label{eq:proofconvergence_strategy1}
    \end{eqnarray}

    Since $\theta^{*} \in int(\Theta)$, there exists $\iota_1 > 0$ such that $\bar \Theta := \{\theta \in \Theta: |\theta - \theta^{*}|_1 \le \iota_1\} \subset int(\Theta)$.

    Apply Lemma \ref{lem:analytical_properties}(ii) with $p = q$ and Lemma \ref{lemma:dependence} applied to $M = 2+\frac{a'}{4}$, $a = \frac{a'}{4}$, we obtain corresponding $C > 0$, $\rho \in (0,1)$, $0 < \iota < \iota_1$ such that the statements of the Lemmata hold true.

    We now use a simple chaining argument. Let $\Theta_{n}$ be a discretization of $\Theta \subset \IR^{r+s+1}$ such that for each $\theta \in \Theta$ there exists some $\theta' \in \Theta_n$ with $|\theta - \theta'|_1 \le n^{-1}$. 
    
    We conclude that for $1 \le m \le n$, it holds that
    \begin{eqnarray}
        &&\IP\Big( \sup_{|\theta - \theta^{*}|_1 < \iota}|S_m(\theta)| > \frac{Q (n\log(n)^3)^{1/2}}{d}\Big)\nonumber\\
        &\le& \IP\Big(\sup_{\theta \in  \Theta_{n}, |\theta - \theta^{*}|_1 < \iota}|S_m(\theta)| > \frac{Q(n\log(n)^3)^{1/2}}{2d}\Big)\nonumber\\
        &&\quad\quad + \IP\Big(\sup_{\theta,\theta' \in \Theta, |\theta - \theta'|_1 \le n^{-1},|\theta - \theta^{*}|_1 < \iota,|\theta' - \theta^{*}|_1 < \iota}|S_m(\theta) - S_m(\theta')| > \frac{Q(n\log(n)^3)^{1/2}}{2d}\Big)\nonumber\\
        &=:& I_n + II_n.\label{eq:proof_strategy_emp}
    \end{eqnarray}
    
    By Lemma \ref{lemma:dependence} applied to $M = 2+\frac{a'}{4}$, $a = \frac{a'}{4}$, we have $\Delta^{\sup_{|\theta - \theta^{*}|_1 < \iota}|g(Z,\theta)|}_{M}(k) = O(\rho^k)$ and $\sup_{|\theta - \theta^{*}|_1 < \iota} \Delta^{g(Z,\theta)}_{M}(k) = O(\rho^k)$. Let $\alpha = \frac{1}{2}$. Then
    \begin{eqnarray*}
        W_{M,\alpha} &:=& \| \sup_{|\theta - \theta^{*}|_1 < \iota}|g(Z,\theta)| \|_{M,\alpha} = \sup_{m \ge 0} (m+1)^{\alpha}\Delta_{M}^{\sup_{|\theta - \theta^{*}|_1 < \iota}|g(Z,\theta)|}(m) < \infty,
    \end{eqnarray*}
    and
    \begin{eqnarray*}
        W_{2,\alpha} &:=& \sup_{|\theta - \theta^{*}|_1 < \iota}\|g(Z,\theta)\|_{2,\alpha} = \sup_{m\ge 0}(m+1)^{\alpha} \sup_{|\theta - \theta^{*}|_1 < \iota}\Delta_{2}^{g(Z,\theta)}(m) < \infty.
    \end{eqnarray*}
    Note that $l = 1 \wedge \log \# \Theta_{n} \le (r+s+1) \log(n)$ and $Q n^{1/2} \log(n)^{3/2} \geq \sqrt{m l}W_{2,\alpha}+m^{1/M} l^{3/2}W_{M,\alpha} \gtrsim m^{1/2} \log(m)^{1/2} + m^{1/M} \log(m)^{3/2}$ for $Q$ large enough.
    
    By applying Theorem 6.2 of \cite{zhangwu2017} with $q = M$ to $(g(Z_i,\theta))_{\theta \in \tilde\Theta_n, |\theta - \theta^{*}|_1 < \iota}$, we have with some constants $C_{\alpha} > 0$:
     \begin{eqnarray}
       I_n &=&\IP\Big(\sup_{\theta \in \Theta_{n}, |\theta - \theta^{*}|_1 < \iota}|S_m(\theta)| > \frac{Q(n\log(n)^3)^{1/2}}{2d}\Big)\nonumber\\
        &\le& \frac{C_{\alpha}m\cdot l^{M/2}W_{M,\alpha}^{M}}{(Q/2d)^{M}(n^{1/2}\log(n)^{3/2})^{M}}+C_{\alpha} \exp \Big( -\frac{C_{\alpha}(Q/2d)^2 n\log(n)^3}{m W_{2,\alpha}^2}\Big)\nonumber\\
        &\le& O(m\cdot n^{-\frac{M}{2}} + n^{-2}),\label{eq:proof_strategy_emp2}
    \end{eqnarray}
    for $Q$ large enough, since $d \le \log_2(n)+1$ and $m \le n$.
    
    Since $g(Z_i,\theta) = \tilde g_{\theta^{*}}(\zeta_i,Y_i,\theta)$ and $g(Z_i,\theta') = \tilde g_{\theta^{*}}(\zeta_i,Y_i,\theta')$, we have with Lemma \ref{lem:analytical_properties}(ii):
    \begin{eqnarray*}
        &&\sup_{\theta,\theta' \in \tilde\Theta, |\theta - \theta'|_1 \le n^{-1},|\theta - \theta^{*}|_1 < \iota,|\theta' - \theta^{*}|_1 < \iota}|g(Z_i,\theta) - g(Z_i,\theta')|\\
        &\le& C(1 + |Y_i|_{(\rho^j)_j,p}^p)(1+\zeta_i^2) n^{-1}.
    \end{eqnarray*}
    Thus
    \begin{eqnarray*}
        && \Big\| \sup_{\theta,\theta' \in \Theta, |\theta - \theta'|_1 \le n^{-1},|\theta - \theta^{*}|_1 < \iota,|\theta' - \theta^{*}|_1 < \iota}\Big|\sum_{i=1}^{m}\big\{\IE_0 g(Z_i,\theta)- \IE_0 g(Z_i,\theta')\big\}\Big|\Big\|_1\\
        &\le& 2C (1 + \| |Y_0|_{(\rho^j)_j,p}^p \|_1) (1 + \IE \zeta_0^2)\frac{m}{n}\\
        &\le& 2C (1 + \frac{D^p}{1-\rho})(1 + \IE \zeta_0^2)\frac{m}{n} = O(\frac{m}{n}).
    \end{eqnarray*}
    With Markov's inequality, we therefore obtain
    \begin{equation}
        II_n \le \frac{2\tilde C m}{(Q/2d) n^{3/2}\log(n)^{3/2}}.\label{eq:proof_strategy_emp3}
    \end{equation}
    Inserting \reff{eq:proof_strategy_emp2} and \reff{eq:proof_strategy_emp3} into \reff{eq:proof_strategy_emp} and then into \reff{eq:proofconvergence_strategy1}, we obtain with some constant $\tilde C > 0$:
    \begin{eqnarray*}
        && \IP\Big(\sup_{|\theta - \theta^{*}|_1 < \iota}\sup_{r \in [0,1]}\Big|\sum_{i=1}^{\lfloor nr\rfloor}\big\{g(X_i^2,Y_i,\theta) - \IE g(X_i^2,Y_i,\theta)\big\}\Big| > Q (n\log(n)^3)^{1/2}\Big)\\
        &\le& \tilde C \sum_{i=0}^{d-1}2^{d-i}\cdot \Big(2^i \cdot n^{-M/2} + n^{-2} + 2^i\cdot n^{-3/2} \log(n)^{1/2}\Big)\\
        &\le& \tilde C d n\cdot \Big(n^{-M/2} + n^{-2} + n^{-3/2}\log(n)^{1/2}\Big) \to 0,
    \end{eqnarray*}
    showing the assertion.
\end{proof}

\subsection{Analytical properties of the likelihood}

For the following results, we derive some analytical properties of the likelihood we use. This allows us to separate analytical and stochastic treatment. For $p > 0$, some sequence $(y_j)_{j\in\IN}$ of real numbers and some sequence $(\chi_j)_{j\in\IN}$ of nonnegative real numbers, define the weighted seminorm
\[
    |y|_{\chi,p} := \Big(\sum_{j=1}^{\infty}\chi_j |y_j|^p\Big)^{1/p}.
\]
Later, we will plug in $x = X_i$ and $y = Y_i$ into $\ell(x,y,\theta)$ and its derivatives. To make use of all connections between $x$,$y$, define $R_{\zeta}(y,\theta) := \zeta^2 \sigma(y,\theta)^2$, and 
\[
    \tilde \ell_{\tilde \theta}(\zeta,y,\theta) := \ell(R_{\zeta}(y,\tilde \theta),y,\theta).
\]
In the following Lemma \ref{lem:analytical_properties}(ii), we collect some analytical properties of $\tilde \ell_{\tilde \theta}$ to calculate functional dependence measures of $\ell(X_i^2,Y_i,\theta)$. The bounds in (iii) will be used to show that the truncated likelihood $\ell(X_i^2,Y_i^c,\theta)$ is near to $\ell(X_i^2,Y_i,\theta)$; for this argument we cannot use the connection between $X_i^2$ and $Y_i$.

\begin{lemma}\label{lem:analytical_properties} $\theta \mapsto \sigma(y,\theta)$ and $\theta \mapsto \ell(x,y,\theta)$ are  three times continuously differentiable. Let $\bar \Theta \subset int(\Theta)$ be a compact subset. Then for any $p > 0$, there exists $\iota > 0$ and $C > 0$, $\rho \in (0,1)$ such that (component-wise),
    \begin{enumerate}
        \item[(i)] for $l = 0,1,2,3$:
        \[
            \sup_{\theta \in \bar \Theta}\frac{|\nabla_{\theta}^l(\sigma(y,\theta)^2)|}{\sigma(y,\theta)^2} \le C (1 + |y|_{(\rho^j)_j,p}^p), \quad\quad \sup_{\theta,\tilde \theta \in \bar \Theta, |\theta - \tilde \theta|_1 < \iota}\frac{\sigma(y,\tilde \theta)^2}{\sigma(y,\theta)^2} \le C (1 + |y|_{(\rho^j)_j,p}^p).
        \]
        \item[(ii)] for $l = 0,1,2$,
        \begin{eqnarray*}
        &&\sup_{\theta, \tilde \theta \in \bar\Theta, |\theta - \tilde \theta|_1 < \iota}|\nabla_{\theta}^l \tilde \ell_{\tilde \theta}(\zeta,y,\theta) - \nabla_{\theta}^l \tilde \ell_{\tilde \theta}(\zeta,y',\theta)|\\
        &\le& C (1 +  |y|_{(\rho^j)_j,2p}^{2p} + |y'|_{(\rho^j)_j,2p}^{2p})|y-y'|_{(\rho^{j})_j,p}^p (1+\zeta^2).
    \end{eqnarray*}
    and
    \[
        \sup_{\theta, \theta', \tilde\theta \in \bar\Theta, |\theta - \tilde \theta|_1 < \iota, |\theta' - \tilde \theta|_1 < \iota}\frac{|\nabla_{\theta}^l \tilde \ell_{\tilde \theta}(\zeta,y,\theta) - \nabla_{\theta}^l \tilde \ell_{\tilde \theta}(\zeta,y,\theta')|}{|\theta - \theta'|_1} \le C(1+|y|_{(\rho^j)_j,p}^p)(1+\zeta^2).
    \]
    \item[(iii)] for $l = 1,2$,
    \[
        \sup_{\theta \in \Theta}|\nabla_{\theta}^l \ell(x,y,\theta) - \nabla_{\theta}\ell(x,y',\theta)| \le C (1 +  |y|_{(\rho^j)_j,1}^{l+1} + |y'|_{(\rho^j)_j,1}^{l+1})(1+|x|)\cdot |y-y'|_{(\rho^{j})_j,1}.
    \]
    \end{enumerate}
\end{lemma}
\begin{proof}[Proof of Lemma \ref{lem:analytical_properties}]
    (i) From Proposition \reff{prop:existencegarch}(iii) we obtain that the following explicit representation holds, where $F(y,\theta) := (\alpha_0 + \sum_{j=1}^{r}\alpha_j y_j,0,...,0)'$:
\begin{equation}
    \sigma(y,\theta)^2 = \sum_{k=0}^{\infty}\big(B(\theta)^k F(y_{k\rightarrow},\theta)\big)_1,\label{garch:sigma_explicit}
\end{equation}
where $y_{k\to} = (y_{k+1},y_{k+2},...)$. 
We conclude that
\begin{eqnarray}
    \sigma(y,\theta)^2 &=& \alpha_0 \sum_{k=0}^{\infty}(B(\theta))_{11} + \sum_{j=1}^{r}\alpha_j \sum_{k=0}^{\infty}(B(\theta)^k)_{11} y_{k+j}\nonumber\\
    &\overset{k'=k+j}{=}& \alpha_0 \sum_{k=0}^{\infty}(B(\theta))_{11} + \sum_{k'=1}^{\infty}\Big(\sum_{j=1}^{r}\alpha_j (B(\theta)^{k'-j})_{11}\Ii_{k' \ge j}\Big) y_{k'}\nonumber\\
    &=:& c_0(\theta) + \sum_{k'=1}^{\infty}c_{k'}(\theta) y_j.\label{garch:sigma_explicit2}
\end{eqnarray}
From Proposition \reff{prop:existencegarch}(iii) we obtain that $c_j(\theta) \ge 0$ satisfies
    \begin{equation}
        \sup_{\theta\in\Theta}|c_k(\theta)| \le C \cdot \rho^k\label{eq:garchgeometricdecaycoefficients}
    \end{equation}
    with some $\rho \in (0,1)$, $C > 0$ and $c_0(\theta) \ge \sigma_{min}^2 > 0$ (due to $\alpha_0 \ge \alpha_{min} > 0$).\\
    Furthermore we conclude that $\sigma(y,\theta)^2$ is three times continuously differentiable w.r.t. $\theta$ with
    \begin{equation}
        \nabla_{\theta}^k(\sigma(y,\theta)^2) = \nabla_{\theta}^k c_0(\theta) + \sum_{k=1}^{\infty}\nabla_{\theta}^k c_k(\theta)\cdot y_{k},\quad\quad k \in \{0,1,2,3\},\label{eq:explicitsigmarepresentation}
    \end{equation}
    where $(\nabla_{\theta}^k c_k(\theta))_{k}$ is still geometrically decaying with $\sup_{\theta \in \Theta}|\nabla_{\theta}^l c_l(\theta)|_{\infty} \le C \cdot \rho^k$, say (enlarge $C > 0,\rho \in (0,1)$ if necessary).\\
    
    In the following we make use of some arguments that were already used in \cite{garch2004}. Note that for $j = 0,...,r$, we have $\partial_{\alpha_j}F(y,\theta) \le \frac{1}{\alpha_j}F(y,\theta)$ and thus
    \begin{equation}
        \partial_{\alpha_j}c_k(\theta) \le \frac{1}{\alpha_j}c_k(\theta).\label{eq:bound_derivative1}
    \end{equation}
    %\[
     %   \frac{\partial_{\alpha_j}(\sigma(y,\theta)^2)}{\sigma(y,\theta)^2} \le \frac{1}{\alpha_j}.
    %\]
    For $j = 1,...,s$, we have ('$\le$' is meant component-wise)
    \[
        \partial_{\beta_j}(B(\theta)^k) = \sum_{i=1}^{k}B(\theta)^{i-1} (\partial_{\beta_j}B(\theta)) B(\theta)^{k-i} \le \frac{1}{\beta_j}k B(\theta)^k.
    \]
    since $\partial_{\beta_j}B(\theta) \le \frac{1}{\beta_j}B(\theta)$. We therefore obtain
    \begin{equation}
        \partial_{\beta_j}c_k(\theta) \le \frac{1}{\beta_j}k \cdot c_k(\theta).\label{eq:bound_derivative2}
    \end{equation}
    From \reff{eq:bound_derivative1} and \reff{eq:bound_derivative2} we obtain the inequalities
    \[
        \partial_{\theta_j}c_k(\theta) \le \frac{k+1}{\theta_j}c_k(\theta).
    \]
    Similar argumentations lead to the bounds for higher order derivatives (cf. also \cite{garch2004}):
    \[
        \partial_{\theta_{j_1}}\partial_{\theta_{j_2}}c_k(\theta) \le \frac{(k+1)^2}{\theta_{j_1} \theta_{j_2}}c_k(\theta), \quad\quad \partial_{\theta_{j_1}}\partial_{\theta_{j_2}}\partial_{\theta_{j_3}}c_k(\theta) \le \frac{(k+1)^3}{\theta_{j_1} \theta_{j_2}\theta_{j_3}}c_k(\theta).
    \]
    
    %\textcolor{blue}{
   %\[
%        \partial_{\theta_{j_1}}\partial_{\theta_{j_2}}c_k(\theta) \le \frac{(k+1)^2}{\theta_{j_1} \theta_{j_2}}c_k(\theta), \quad\quad \partial_{\theta_{j_1}}\partial_{\theta_{j_2}}\partial_{\theta_{j_3}}c_k(\theta) \le \frac{(k+1)^3}{\theta_{j_1} \theta_{j_2}\theta_{j_3}}c_k(\theta).
%    \]
%    }
    
    If $\bar\Theta \subset int(\tilde \Theta)$ is some compact subspace, we therefore obtain with $C_1 := \max\{\frac{1}{\theta_j}:j=1,...,r+s+1,\theta \in \bar\Theta\}$ for arbitrary small $p > 0$:
    \begin{eqnarray*}
        \frac{\partial_{\theta_j}(\sigma(y,\theta)^2)}{\sigma(y,\theta)^2} &\le& C_1 \frac{\sum_{k=0}^{\infty}(k+1) c_k(\theta)}{\sum_{k=0}^{\infty}c_k(\theta)}\\
        &\le& \frac{C_1 c_0(\theta)}{\sigma_{min}^2} + C_1\sum_{k=1}^{\infty}(k+1) \frac{c_k(\theta) y_k}{c_0(\theta) + c_k(\theta)y_k}\\
        &\le& \frac{C_1 c_0(\theta)}{\sigma_{min}^2} + \sum_{k=1}^{\infty}(k+1) \Big(\frac{c_k(\theta)}{c_0(\theta)}\Big)^{p} y_k^p,
    \end{eqnarray*}
    where we have used $\frac{x}{1+x} \le x^s$ in the last inequality. Since $c_k(\theta)^s \le C^s (\rho^{s})^k$, we can find $\tilde C > 0, \tilde \rho \in (0,1)$ such that
    \[
        \sup_{\theta \in \bar \Theta}\frac{|\partial_{\theta_j}(\sigma(y,\theta)^2)|}{\sigma(y,\theta)^2} \le \tilde C (1 + |y|_{(\tilde \rho^j)_j,p}^{p}),
    \]
    and similarly for the higher order derivatives (component-wise):
    \[
        \sup_{\theta \in\bar\Theta}\frac{|\nabla_{\theta}^l(\sigma(y,\theta)^2)|}{\sigma(y,\theta)^2} \le \tilde C (1 + |y|_{(\tilde \rho^j)_j,p}^{p}), \quad l = 1,2,3.
    \]
    For $\theta, \tilde \theta \in \bar\Theta$ and arbitrary small $p > 0$, choose $\delta > 0$ such that $\bar \rho := (1+\delta)\rho^p < 1$. Then choose $\iota > 0$ such that $|\theta - \tilde \theta|_1 < \iota$ implies (component-wise) $B(\tilde \theta) \le (1+\delta) B(\theta)$. For $|\theta - \tilde \theta| < \iota$, it then holds that $c_k(\tilde \theta) \le (1+\delta)^k c_k(\theta)$. We conclude that
    \begin{eqnarray*}
        \frac{\sigma(y,\tilde \theta)^2}{\sigma(y,\theta)^2} &\le& \frac{c_0(\tilde \theta)}{\sigma_{min}^2} + \sum_{k=1}^{\infty}\frac{c_k(\tilde \theta)y_k}{c_0(\theta) + c_k(\theta) y_k}\\
        &\le& \frac{c_0(\tilde \theta)}{\sigma_{min}^2} +  \sum_{k=1}^{\infty}\frac{c_k(\tilde \theta)}{c_k(\theta)}\cdot \Big(\frac{c_k(\theta)}{c_0(\theta)}\Big)^p y_k^p\\
        &\le& \frac{c_0(\tilde \theta)}{\sigma_{min}^2} + \frac{C^p}{\sigma_{min}^{2p}}\sum_{k=1}^{\infty}((1+\delta) \rho^{p})^{k}y_k^p.
    \end{eqnarray*}
    We conclude that there exists $\bar C > 0$, $\bar \rho \in (0,1)$ such that
    \[
        \sup_{\theta,\tilde \theta \in \bar \Theta, |\theta - \tilde \theta|_1 < \iota}\frac{\sigma(y,\tilde \theta)^2}{\sigma(y,\theta)^2} \le \bar C (1 + |y|_{(\bar \rho^j)_j,p}^p).
    \]
    (ii) From the differentiability of $\theta \mapsto \sigma(y,\theta)$ we obtain that $\theta \mapsto \ell(x,y,\theta)$ is three times continuously differentiable and
\begin{eqnarray}
    \ell(x,y,\theta) &=& \frac{1}{2}\Big(\frac{x}{\sigma(y,\theta)^2} + \log(\sigma(y,\theta)^2)\Big),\label{eq:likelihood_der1}\\
    \nabla_{\theta}\ell(x,y,\theta) &=& \frac{\nabla_{\theta}(\sigma(y,\theta)^2)}{2\sigma(y,\theta)^2}\Big(1-\frac{x}{\sigma(y,\theta)^2}\Big),\label{eq:likelihood_der2}\\
    \nabla_{\theta}^2 \ell(x,y,\theta) &=& \Big[-\frac{\nabla_{\theta}(\sigma(y,\theta)^2) \nabla_{\theta}(\sigma(y,\theta)^2)'}{2 \sigma(y,\theta)^4} +\frac{\nabla_{\theta}^2(\sigma(y,\theta)^2)}{2\sigma(y,\theta)^2} \Big]\Big(1-\frac{x}{\sigma(y,\theta)^2}\Big)\nonumber\\
    &&\quad\quad\quad + \frac{\nabla_{\theta}(\sigma(y,\theta)^2) \nabla_{\theta}(\sigma(y,\theta)^2)'}{2\sigma(y,\theta)^4}\cdot \frac{x}{\sigma(y,\theta)^2}.\label{eq:likelihood_der3}
\end{eqnarray}

%{\color{blue}
%\begin{equation}
%    \nabla_{\theta}\ell(x,y,\theta) = \frac{\nabla_{\theta}(\sigma(y,\theta))}{\sigma(y,\theta)}\Big(1-\frac{x}{\sigma(y,\theta)^2}\Big),
%    \end{equation}

%}
For the corresponding quantity $\tilde \ell_{\tilde \theta}$ we obtain
\[
    \nabla_{\theta}\tilde\ell_{\tilde \theta}(\zeta,y,\theta) = \frac{\nabla_{\theta}(\sigma(y,\theta)^2)}{2\sigma(y,\theta)^2}\Big(1-\frac{\sigma(y,\tilde \theta)^2}{\sigma(y,\theta)^2}\zeta^2\Big).
\]
By (i), we obtain that for $p > 0$, there exist constants $\iota > 0$, $C_2 > 0$, $\rho_2 \in (0,1)$ such that (component-wise):
\begin{equation}
    \sup_{|\theta - \tilde \theta|_1 < \iota}|\nabla_{\theta}\tilde \ell_{\tilde \theta}(\zeta,y,\theta)| \le C_2(1 + |y|_{(\rho_2^j)_j,{p/2}}^{p/2})(1 + (1 + |y|_{(\rho_2^j)_j,{p/2}}^{p/2}) \zeta^2).\label{eq:analysis_intermediate}
\end{equation}
By using
\begin{eqnarray}
    |y|_{(\rho_2^j)_j,p/2}^{p/2} &\le& \sum_{j=1}^{\infty}\rho_2^{j/2}\cdot \rho_2^{j/2}y_j^{p} \le (\sum_{j=1}^{\infty}\rho_2^j)^{1/2}(\sum_{j=1}^{\infty}\rho_2^{j}y_j^{p})^{1/2}\nonumber\\
    &=& (1-\rho_2)^{-1/2}|y|_{(\rho_2^j)_j,p}^{p/2},\label{eq:standardestimation}
\end{eqnarray}
we can obtain the more compact form
\[
    \sup_{|\theta - \tilde \theta|_1 < \iota}|\nabla_{\theta}\tilde \ell_{\tilde \theta}(\zeta,y,\theta)| \le C_3(1 + |y|_{(\rho_2^j)_j,p}^p)(1+ \zeta^2).
\]
with some new constant $C_3 > 0$. Due to the similar structure, we can use similar techniques to obtain (component-wise):
\begin{equation}
    \sup_{|\theta - \tilde \theta|_1 < \iota}|\nabla_{\theta}^l\tilde \ell_{\tilde \theta}(\zeta,y,\theta)| \le C_3(1 + |y|_{(\rho_2^j)_j,p}^p)(1+ \zeta^2)=: M_p(y,\zeta), \quad\quad l = 1,2,3.\label{eq:likelihood_upperbound}
\end{equation}

    From \reff{eq:explicitsigmarepresentation} we deduce that (component-wise) for $l = 0,1,2$ with some constant $C_4 > 0$, uniformly in $\theta,\theta' \in \bar\Theta$:
    \begin{eqnarray}
        |\nabla_{\theta}^l(\sigma(y,\theta)^2) - \nabla_{\theta}^l(\sigma(y',\theta)^2)| &\le& C_4 |y - y'|_{(\rho^j)_j,1},\label{eq:sigmagarch_xdiff}
        %|\nabla_{\theta}^l(\sigma(y,\theta)^2) - \nabla_{\theta}^l(\sigma(y,\theta')^2)| &\le& |\theta - \theta'|_1 \cdot \sup_{\theta \in \Theta}|\nabla_{\theta}^{l+1}(\sigma(y,\theta)^2)|_{\infty}\nonumber\\
        %&\le& C_4 |\theta - \theta'|_1 \cdot |y|_{(\rho^j)_j,1}.\label{eq:sigmagarch_thetadiff}
    \end{eqnarray}
    
    By using $|\frac{1}{\sigma(y,\theta)^2}-\frac{1}{\sigma(y,\theta)^2}| \le \frac{1}{\sigma_{min}^4}|\sigma(y,\theta)^2 - \sigma(y',\theta)^2|$ and the very rough bounds $\sigma(y,\theta)^2 \ge \sigma_{min}^2$, \reff{eq:sigmagarch_xdiff} and \reff{eq:explicitsigmarepresentation}, we obtain (component-wise) with some constant $C_5 > 0$:
    \[
        \sup_{\theta \in \bar \Theta}|\nabla_{\theta} \tilde \ell_{\tilde \theta}(\zeta,y,\theta) - \nabla_{\theta} \tilde \ell_{\tilde \theta}(\zeta,y',\theta)| \le C_5(1 + |y|_{(\rho^j)_j,1} + |y'|_{(\rho^j)_j,1})^2 |y - y'|_{(\rho^j)_j,1}(1+\zeta^2)
    \]
    Similar results can be obtained for higher derivatives (component-wise), $l = 1,2$:
    \begin{eqnarray}
        &&\sup_{\theta \in \bar\Theta}|\nabla_{\theta}^l \tilde \ell_{\tilde \theta}(\zeta,y,\theta) - \nabla_{\theta} \tilde \ell_{\tilde \theta}(\zeta,y',\theta)|\nonumber\\
        &\le& C_5(1 + |y|_{(\rho^j)_j,1} + |y'|_{(\rho^j)_j,1})^{1+l} |y - y'|_{(\rho^j)_j,1}(1+\zeta^2) =: N_l(y,y',\zeta).\label{eq:likelihood_difference_x}
    \end{eqnarray}
    %For differences w.r.t. $\theta$, we use \reff{eq:sigmagarch_thetadiff} instead of \reff{eq:sigmagarch_xdiff} and obtain with some constant $C_6 > 0$:
    %\begin{equation}
     %   \sup_{\theta,\theta' \in \bar\Theta, \theta \not=\theta'}\frac{|\nabla_{\theta}^l \tilde \ell_{\tilde \theta}(\zeta,y,\theta) - \nabla_{\theta} \tilde \ell_{\tilde \theta}(\zeta,y',\theta)|}{|\theta - \theta'|_1} \le C_6(1 + |y|_{(\rho^j)_j,1})^{2+l}, \quad l = 1,2.\label{eq:likelihood_difference_theta}
    %\end{equation}
    Using \reff{eq:likelihood_upperbound} and \reff{eq:likelihood_difference_x}, we have for $l = 1,2$ and arbitrary small $p' > 0$ (use $\min\{1,x\} \le x^{p'}$):
    \begin{eqnarray*}
        && \sup_{\theta \in \bar\Theta}|\nabla_{\theta}^l \tilde \ell_{\tilde \theta}(\zeta,y,\theta) - \nabla_{\theta} \tilde \ell_{\tilde \theta}(\zeta,y',\theta)|\\
        &\le& \min\{ M_p(y,\zeta) + M_p(y',\zeta), N_l(y,y',\zeta)\}\\
        &=& \{M_p(y,\zeta) + M_p(y',\zeta)\}\min\{1, \frac{N_l(y,y',\zeta)}{M_p(y,\zeta) + M_p(y',\zeta)}\}\\
        &\le& \{M_p(y,\zeta) + M_p(y',\zeta)\} \Big(\frac{N_l(y,\zeta)}{M_p(y,\zeta) + M_p(y',\zeta)}\Big)^{p'}\\
        &=& \{M_p(y,\zeta) + M_p(y',\zeta)\}^{1-p'} N_l(y,\zeta)^{p'}.
    \end{eqnarray*}
    Choosing $p' \in (0, \min\{1,p'(1+l)\}$, we obtain
    \begin{eqnarray*}
        \{M_p(y,\zeta) + M_p(y',\zeta)\}^{1-p'} &\le& C_3^{1-p'}(1+|y|_{(\rho_2^j)_j,p}^p+|y'|_{(\rho_2^j)_j,p}^p)(1+\zeta^2)^{1-p'},\\
        N_l(y,y',\zeta)^{p'} &\le& C_5^{p'}(1 + |y|_{(\rho^{pj}),p}^p+|y'|_{(\rho^{pj}),p}^p) |y-y'|_{(\rho^{pj})_j,p}^p (1 + \zeta^2)^{p'}.
    \end{eqnarray*}
    With \reff{eq:standardestimation}, $\rho_3 := \max\{\rho_2,\rho^{p}\}$ and some constant $C_6 > 0$
    \[
        \sup_{\theta, \tilde \theta \in \bar\Theta, |\theta - \tilde \theta|_1 < \iota}|\nabla_{\theta}^l \tilde \ell_{\tilde \theta}(\zeta,y,\theta) - \nabla_{\theta}^l \tilde \ell_{\tilde \theta}(\zeta,y',\theta)| \le C_6 (1 +  |y|_{(\rho_3^j)_j,2p}^{2p}+|y'|_{(\rho_3^j)_j,2p}^{2p})|y-y'|_{(\rho^{j}_3)_j,p}^p (1+\zeta^2).
    \]
    By using \reff{eq:likelihood_upperbound} and the mean value theorem, we obtain for $l = 1,2$:
    \begin{eqnarray*}
        &&\sup_{\theta, \theta', \tilde\theta \in \bar\Theta, |\theta - \tilde \theta|_1 < \iota, |\theta' - \tilde \theta|_1 < \iota}\frac{|\nabla_{\theta}^l \tilde \ell_{\tilde \theta}(\zeta,y,\theta) - \nabla_{\theta}^l \tilde \ell_{\tilde \theta}(\zeta,y,\theta')|}{|\theta - \theta'|_1}\\
        &\le& \sup_{|\bar\theta - \tilde \theta|_1 < \iota}|\nabla_{\theta}^{l+1}\tilde \ell_{\tilde \theta}(\zeta,y,\bar \theta)|_{\infty} \le M_p(y,\zeta),
    \end{eqnarray*}
    giving the result.\\
    (iii) Using the representations \reff{eq:likelihood_der2}, \reff{eq:likelihood_der3} and the inequalities \reff{eq:sigmagarch_xdiff}, \reff{eq:explicitsigmarepresentation} and $\sigma(y,\theta)^2 \ge \sigma_{min}^2$, this is an immediate consequence.
\end{proof}

\bibliographystyle{plain}
\bibliography{references}
\end{document}